\documentclass[journal]{IEEEtran}
\usepackage{amsmath}
\usepackage{amsthm}
\usepackage{amsfonts}
\usepackage{amssymb}
\usepackage{amscd}
\usepackage{graphicx}
\usepackage{epsfig}
\usepackage{caption}
\usepackage{subcaption}
\usepackage{colortbl}
\usepackage{newlfont}
\usepackage{cite}
\usepackage{ifthen}
\usepackage{times}
\IEEEoverridecommandlockouts
\newtheorem{theorem}{Theorem}
\theoremstyle{definition}

\newtheorem{corollary}{Corollary}
\theoremstyle{remark}

\newcommand{\setxysizemod}{\epsfxsize=2in}
\newcommand{\lb}{\left(}
\newcommand{\rb}{\right)}

\newcommand{\lcb}{\left\{}
\newcommand{\rcb}{\right\}}

\begin{document}

\title{On the Secrecy Capacity of $2$-User Gaussian Z-Interference Channel with Shared Key} 

\author{\IEEEauthorblockN{Somalatha U and Parthajit Mohapatra}\vspace{-1cm}
\thanks{The authors are with the Dept. of Electrical Engineering, Indian Institute of Technology Tirupati 517 506, India (e-mails:\{ee19d002, parthajit\}@iittp.ac.in)}}
\maketitle
\begin{abstract}
In this paper, the role of secret key with finite rate is studied to enhance the secrecy performance of the system when users are operating in interference limited scenarios. To address this problem, a $2$-user Gaussian Z-interference channel with secrecy constraint at the receiver is considered. One of the fundamental problems here is how to use the secret key as a part of the encoding  process. The paper proposes novel achievable schemes, where the schemes differ from each other based on how the key has been used in the encoding process. The first achievable scheme uses a combination of \emph{key rate splitting}, \emph{one-time pad}, \emph{stochastic encoding} and \emph{superposition coding}. In this scheme, one part of the key is used for one-time pad and the remaining part of the key is used for wiretap coding. The encoding is performed such that the receiver experiencing interference can decode some part of the interference without violating the secrecy constraint. As a special case of the derived result, one can obtain the secrecy rate region when the key is completely used for one-time pad or part of the wiretap coding. The second scheme uses the shared key to encrypt the message using one-time pad and in contrast to the previous case no interference is decoded at the receiver. The paper also derives an outer bound on the sum rate and secrecy rate of the transmitter which causes interference.  The main novelty of deriving outer bound lies in the selection of side information provided to the receiver and using the secrecy constraint at the receiver. The derived outer bounds are found to be tight depending on the channel conditions and rate of the key. The scaling behaviour of key rate is also explored for different schemes using the notion of  \emph{secure generalized degrees of freedom (SGDOF)}.  The optimality of different schemes are characterized for some specific cases. The developed results provide new insights on the performance of different achievable schemes under varied channel conditions and rate of the key. The developed results show the importance of key rate splitting in enhancing the secrecy performance of the system when users are operating under interference limited environment. 
\end{abstract}
\section{Introduction}
The upcoming 5G and beyond wireless communication technologies will be a key enabler for the Internet of Things (IoT). These networks are expected to support immense user connections and increasing wireless services. As the wireless communication is prone to attack, ensuring information security over such networks is of paramount importance. The result in \cite{shannon-bell-1949} has shown that one-time pad can achieve perfect secrecy for the point-to-point channel provided the key size is as large as that of the message size. When the shared key is used independent of the channel codes in case of point to point scenario, the performance of the system  is either limited by the rate of the key or the capacity of the channel. In multi-user scenarios, users can have varied knowledge of the key and using the key independent of the channel code can result in sub optimal performance of the system. In general, users cannot use interference cancellation or joint decoding technique to improve their performance as it will violate the secrecy condition. This brings a fundamental question on how to use the secret key as a part of the channel coding to improve the performance of the system in interference limited scenarios. To address this problem, this paper considers a 2-user Gaussian Z-Interference channel (ZIC) with shared key of finite rate. 

The Z-Interference channel (Z-IC)~\cite{liu2004capacity, mohapatra2016secrecy, li2008secrecy} has been studied in existing literature with and without secrecy constraint. In this case, only one of the users causes interference. In practice, the Z-IC model can capture heterogeneous scenarios, where the user which is far from the unintended transmitter does not experience any interference.  Although the Z-IC is simplified model in comparison to the  interference channel, its capacity region is not known fully known~\cite{mohapatra2016secrecy, liu2009capacity} even without secrecy constraint at the receiver, except for some specific cases. The Z-IC model still captures the broadcast and superposition nature of the wireless medium and  study of such model can provide useful insights on system performance. 
\subsection{Related work} 
Information theoretic secrecy has given new insights on designing secure communication schemes using the randomness present in wireless communication system. The work in \cite{wyner-bell-1975} has shown that noise present in communication system can be useful for secure communication.  The study of different fundamental models of wireless communication with secrecy constraints such as broadcast channel and interference channel has given new insights on the performance limit ~\cite{wyner-bell-1975, csiszar1978broadcast, liu2008discrete}. The interference channel which captures the impact of concurrent transmission has been studied with and without secrecy constraint under different settings ~\cite{liu2008discrete, mohapatra-tit-2016, etkin-TIT-2008}. The results related to other communication models with secrecy constraints can be found in~\cite{xu-tifs-2013, fritschek-tifs-2019, fayed-asilomer-2016}. 
 
When users have access to a secret key, the performance of the wireless network depends on how the secret key is used as a part of channel codes and the varied knowledge of the key at the various nodes. The works ~\cite{shannon-bell-1949, yamamoto1997rate, merhav2008shannon, kang2010wiretap, ardestanizadeh2009wiretap} consider a point-to-point system in the presence of an eavesdropper with a shared key between the legitimate nodes under different channel conditions and settings.  The role of the secret key as a part of channel coding for multi-user scenarios has been explored in recent years. The work in ~\cite{schaefer2017secure} considers a 2-user broadcast channel with an external eavesdropper where the transmitter shares key of high rate with respective receivers. The role of the secret key in the interference-limited scenario has been studied ~\cite{sinha2018secrecy} for the 2-user Gaussian interference channel (GIC). When users are operating under interference-limited scenarios, receivers cannot decode the interference to cancel its effect due to secrecy constraint at the receivers. This in turn can reduce the performance of the system significantly. When users have access to shared keys an important question arises how to use the secret key in the encoding and decoding process so that the overall performance of the system can be improved. To answer this question, a $2$-user Gaussian Z interference channel with shared key between transmitter $2$ and receiver $2$ is considered. 

The Z-Interference channel have been studied in the literature with and without secrecy constraints~\cite{liu2004capacity, mohapatra2016secrecy},~\cite{liu2009capacity} and~\cite{salehkalaibar2010capacity}. The capacity region is characterized for a specific class of degraded Z-IC, where the achievable scheme uses a Marton's binning technique \cite{salehkalaibar2010capacity}. The sum capacity of 2-user Gaussian Z-IC  was characterized in~\cite{sason2004achievable}. The capacity region is also characterized for a class of deterministic Z-IC~\cite{gamal1982capacity}. The work in \cite{zhou2012gaussian} considers a Z-IC with a relay link and the achievable scheme uses a combination of Han-Kobayashi scheme and interference forwarding relaying strategy. The capacity region for this model is characterized for the strong interference regime and the asymptotic sum-capacity is characterized for the weak-interference regime.  In~\cite{bagheri2010approximate}, the role of transmitters cooperation is explored for the 2-user Gaussian Z-IC and the achievable scheme uses a combination of  Han-Kobayashi scheme, zero-forcing, and relaying techniques. The role of bidirectional digital communication link between the two receivers is investigated for $2$-user Gaussian Z-IC in~\cite{do2009achievable}.  

The problem of secure communication over 2-user Z-IC with secrecy constraint has been explored in \cite{li2008secrecy}. The capacity region is characterized for the deterministic model and for the Gaussian case the capacity region is characterized within one bit when the direct link is stronger than the interfering link. In~\cite{mohapatra2016secrecy}, the role of transmitter cooperation is studied for Z-interference channel with secrecy constraint at the receiver. The capacity region is characterized for the deterministic model with cooperation and using the insights obtained from the deterministic model, bounds on the secrecy capacity region are obtained for the Gaussian model. The role of shared secret key in enabling secure communication under interference limited scenarios is not well understood from the existing literature. Some initial works in this direction can be found for 2-user Gaussian IC in \cite{sinha2018secrecy}. However, it is not possible to directly use the developed achievable schemes and outer bounds in \cite{sinha2018secrecy} for the Z-IC. 
\subsection{Contributions}
This paper considers a $2$-user Gaussian Z-IC with secrecy constraint, where a key of finite rate is shared between the users who causes interference to the other receiver. One of the fundamental problem arises on how to use the secret key as a part of the channel encoding process. The derivation of the outer bounds are non-trivial as it needs to exploit the secrecy constraint at receiver and  knowledge of the shared key in a judicious manner. The work presents novel bounds on the secrecy capacity region and these bounds perform differently based on the channel strength, rate of the key and how the key has been used as a part of the encoding process. The main contributions of the work are as follows:
\begin{enumerate}
	\item The first achievable scheme is based on the novel idea of key rate splitting, where one part of the key is used for common confidential message and the remaining part of the key is used as a part of wiretap coding. This scheme uses a combination of rate splitting technique, one time pad, stochastic encoding and superposition coding. The novelty of the scheme is that the encoding at the transmitter (which causes interference) is performed such that the unintended receiver can decode some part of the interference and cancel its effect without violating the secrecy criteria. To the best of authors' knowledge, this kind of key rate splitting has not been explored for the interference limited scenario. Using some part of the key, it is possible to send additional message in the wiretap coding. As a special case of the derived result, one can obtain the achievable rate regions for two other cases: when the key is completely used for either encoding the common confidential message\footnote{It is called as common confidential message as the message is decodable at both the receivers.} or part of the wiretap coding. 
	
	\item This work also proposes another achievable scheme, where the key is used to encrypt the message of transmitter~$2$ and this scheme is generally known as one time pad in the literature. The receiver which experiences interference decodes its intended message by treating other user's message as noise. This scheme can be considered as the secure version of  treating interference as noise, which is conventionally used in many wireless systems. 
	
	\item  The work also derives an outer bound on the sum rate by taking account of the secrecy constraint at receiver  and the finite key rate shared between the users. This outer bound is applicable only for the weak/moderate interference regime. Another outer bound is also derived on the secrecy rate of user~$2$ and this bound is found to be tight for the very high interference regime.  The main novelty of deriving the outer bounds lies in the choice of side information provided at the receiver and using secrecy constraint at the receiver in a judicious manner.
	
	\item The work also explores the scaling behaviour of the key rate for different schemes with respect to underlying channel parameters such as SNR and INR for the considered model. To capture this, the notion of secure generalized degrees of freedom (GDOF) is used and the GDOF region for different achievable schemes are characterized in the weak/moderate interference regime. 
\end{enumerate}

The derived results shed insights on the role of shared key in enabling secure communication for interference-limited scenarios. It is found that the achievable scheme based on splitting the key rate performs best among the other schemes considered. As shared randomness between users is a promising theme for secure communication, it is important to understand whether using the key as a part of stochastic encoding is more beneficial or decoding some part of the interference is more useful. The derived GDOF results also help to explore the scaling behaviour of the key rate in terms of SNR or INR to improve the performance of the system, in the weak/moderate interference regime.

\section{System Model}\label{sec:sysmod}
This paper considers a $2$-user Gaussian Z-IC with shared secret key $K$ of finite rate $R_{K}$. The transmitter (Tx-2) which causes interference has access to the key as well as its corresponding receiver (Rx-2) (See Fig.~\ref{fig:ZIC}). The relation between input and output of this model is given below
\begin{align}
 Y_1&= h_{11}X_{1} + h_{21}X_2 + Z_{1}, \nonumber\\
 Y_2& = h_{22}X_{2} + Z_{2}, \label{sysmodel1}
\end{align}

\begin{figure}[h]
\begin{center}
\setxysizemod
  \includegraphics[width=\linewidth]{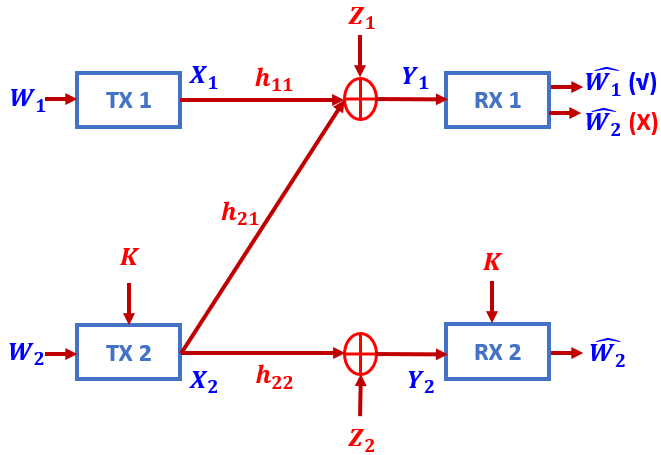}
  \caption{Two user ZIC with secret key.}
  \label{fig:ZIC}
  \end{center}
\end{figure}

where $X_{i}\:(i=1,2)$ represents the real valued transmitted signal which needs to satisfy the average power constraint, i.~e., $E[X_i^2]\leq P_i$. The quantity $Z_{i}\:(i \in \{1,2\})$ represents the real valued additive white Gaussian noise (AWGN)  ($Z_{i} \sim \mathcal{N}(0,1)$) at receiver~$i$.  The channel gain between transmitter~$i$ and receiver~$j$ is denoted by $h_{ij}$. It is also assumed that transmitters and receivers have global channel knowledge.

The transmitter $i$ has a message  $W_i \in \{1, 2, \ldots, 2^{nR_i}\}$ intended for receiver~$i$, but the message of transmitter~$i$ needs to be kept secret from receiver~$j$ $(j\neq i)$. It should be noted that, for the Z-IC, message of transmitter~$1$ is secured due to absence of communication link between transmitter $1$ and receiver $2$. Hence, there is no key shared between transmitter $1$ and receiver $1$. However, the message of transmitter~$2$ is required to be kept secret from receiver~$1$. Both the messages at the transmitters are independent of each other. The shared secret key $K$ is independent of the messages and is  uniformly distributed over the set $\mathcal{K} = \{1, 2, \ldots, 2^{nR_K}\}$. Transmitter~$1$ or receiver~$1$ does not have knowledge of the key.

The weak notion of secrecy is used to measure the secrecy performance of the system. Mathematically, it is defined as $I(W_2;Y_1^n) \leq n \epsilon_n$ where $\epsilon_n \to 0$ as $n \to \infty$ and $n$ corresponds to the length of the codeword \cite{bloch2013strong}. The details of encoding and decoding are presented in the following section. 
\section{Achievable Schemes}\label{achievable schemes}
In this section, achievable schemes are proposed for the considered model and these achievable schemes differ from each other based on how the key has been used in the encoding process. The first achievable scheme is based on the novel idea of key rate splitting, where one  part of the key is used for one-time pad and the remaining part of the key is used in the wiretap coding. The proposed scheme uses a combination of key-rate splitting, one-time pad, stochastic encoding and superposition coding. Another novel aspect of the scheme is that some part of the interference can be cancelled at the unintended receiver without violating the secrecy constraint. As a by product of this derived result, one can obtain the achievable results for two extreme cases: (a) the key is used only for one-time pad and (b) key is used as a part of the stochastic encoding. The paper also proposes another achievable scheme where the key is used to protect the message of transmitter~$2$ using one-time pad. In this case, receiver~$1$ decodes its intended message by treating interference as noise in contrast to the previous scheme where receiver needs to decode some part of the interference. The performance of various schemes depend on how the key has been used as a part of the encoding process, decoding techniques at the receivers, key rate and channel conditions. These schemes are discussed in details in the following sub-sections.
\subsection{Key rate Splitting}\label{achievable scheme:key splitting}
One of the ways to mitigate interference is to decode some part of it or the complete interference based on its strength. When transmitter has confidential data, in general, receiver cannot decode other user's message as it will violate the secrecy constraint. However, when the intended user pairs have access to a key, the transmitter can generate a ciphertext using the key and the ciphertext is encoded such that it remains decodable at both the receivers.  The unintended receiver can determine the ciphertext, but it cannot determine the message as it does not know the key. This can allow receiver~$1$ to decode some part of the interference and cancel its effect without violating the secrecy constraint. This message is termed as \emph{common confidential message} as this is required to be decodable at both the receivers and the rate associated with the message is limited by the channel conditions and the rate of the key used for the encoding.  In addition to this, it may be possible to send an additional message to the receiver~$2$ using stochastic encoding and this message is termed as \emph{private message}. A part of the key can also be used in the wiretap coding and this in turn can help to reduce the loss of rate due to stochastic encoding. This motivates to split the rate of the key into two parts: one part is used to encode the common confidential message and other part is used in the wiretap coding. This is termed as \emph{key rate splitting} in this work. The message at transmitter~$2$ is split into two parts: common confidential message  $(W_{2C})$ and the private message $({W_{2P}^W})$. The private message $({W_{2P}^W})$ is further split into two parts: one part is encoded using stochastic encoding and another part of the message is encoded with the help of key in the wiretap coding.
	
To the best of authors' knowledge, key rate splitting has not been explored in the existing literature for the interference limited scenarios. This also brings an important question on whether it is better to use the entire key for encoding the common confidential message, in the wiretap coding or key rate splitting. One can obtain the achievable secrecy rate regions for these two extreme cases as well from the achievable scheme based on key rate splitting. The achievable result with key rate splitting is stated in the following theorem.
\begin{theorem}\label{theorem-key splitting}
For the $2$-user Gaussian Z-IC with shared key, the scheme based on key rate splitting achieves the following secrecy rate region
\begin{align}
    R_1&\leq I(X_1;Y_1|X_{2C}),\nonumber\\
    R_2&\leq \min\{(I(X_{2C};Y_1|X_1), I(X_{2C};Y_2),\eta R_K\}\nonumber\\
    &\quad+\min\{I(X_{2P};Y_2|X_{2C}), I(X_{2P};Y_2|X_{2C})\nonumber\\
    &\qquad\qquad+ (1-\eta)R_K - R_{2P}' \},\nonumber\\
    R_1+R_2&\leq I(X_1,X_{2C};Y_1) + \min\{(X_{2P}; Y_2|X_{2C}),\nonumber\\
    &\quad I(X_{2P};Y_2|X_{2C}) + (1-\eta)R_K- R_{2P}'\}, \label{achievable secrecy rate region for key splitting} \\
\text{where}\quad R_{2P}' & \geq I(X_{2P};Y_1 | X_1, X_{2C}) \text{ and } 0 \leq \eta \leq 1. \label{equivocation computation}
\end{align}
\end{theorem}
\begin{proof}
The proof involves error probability analysis at both the receivers and equivocation computation to know how much rate the transmitter $2$ has to sacrifice to keep the private message of its confidential. Due to the asymmetry of the model considered in this work, transmitter~$1$ uses deterministic encoding where the message is encoded using a point-to-point code, and transmitter~$2$ uses a combination of deterministic encoding and stochastic encoding. The decoding at receiver~$1$ is based on jointly typical set decoding and successive decoding is used at receiver~$2$. The details are as follows.

\textbf{\emph{Codebook generation:}} 
Transmitter~$1$ generates a codebook $\mathcal{C}_1$ containing $2^{nR_1}$ independent and identically distributed (i.i.d.) random codewords of length $n$ according to $\mathcal{N}(0, P_1)$. Similarly, transmitter $2$ generates $2^{nR_{2C}}$ codewords for the common confidential message. For the private message, transmitter~$2$ generates $2^{n(R_{2P}+R_{2P}'+{R_{2P}^\oplus})}$ codewords and these codewords are grouped into $2^{n R_{2P}}$ bins where each bin contains $2^{n R_{2P}'}$ sections, and each of which has $2^{n R_{2P}^\oplus}$ codewords. These codebooks are made available to both the receivers.
  
\textbf{\emph{Encoding}}  
Transmitter $1$ uses the classical channel encoder to encode the message $W_1$. Transmitter $2$ splits the message $W_2$ into common confidential message $W_{2C}\in \mathcal{W}_{2C} = \{1,2,...,2^{nR_{2C}}\}$ and private message $W_{2P}^W \in \mathcal{W}_{2P}^W = \{1,2,...,2^{nR_{2P}^W}\}$. The secrecy of $W_{2C}$ can be guaranteed by using the key of rate $\eta R_K$. This part of the key is used to generate the encrypted message $W_{2C}'$ which is encoded using the codebook for the common confidential message. Let ${W_{2P}^W}=(W_{2P},{W_{2P}^\oplus})$, where the message $W_{2P}$ is transmitted  using  stochastic encoding scheme while the message $W_{2P}^\oplus$ is sent as a part of the wiretap coding using the remaining part of the key $(1-\eta)R_K$. The message $W_{2P}$ selects the bin among $2^{nR_{2P}}$ bins. In that bin, one of the sections is picked uniformly at random and in that section the particular codeword is chosen among $2^{nR_{2P}^\oplus}$ based on the encrypted message of ${W_{2P}^\oplus}$. Index the codewords in each section from  1  to  $2^{nR_{2P}^\oplus}$, so that the codewords in the jth section of the ith sub-codebook can be denoted as $X_{2P, C_{ij}}(1), \ldots, X_{2P, C_{ij}}(2^{nR_{2P}^\oplus})$.

Finally, the transmitter~2 superimposes the codewords corresponding to the common confidential message and private message and sent over the channel as given below
\begin{align}
    X_2^n( W_{2C}',W_{2P}, W_{2P}', W_{2P}^\oplus) & = X_{2C}^n(W_{2C}') \nonumber\\
    &\qquad + X_{2P}^n(W_{2P}, W_{2P}', W_{2P}^\oplus)
\end{align}
A schematic representation of the encoding process at transmitter~$2$ is shown in Fig.~\ref{fig:encoding process}.
\begin{figure}[h]
	\begin{center}
		\setxysizemod
		\includegraphics[width=\linewidth]{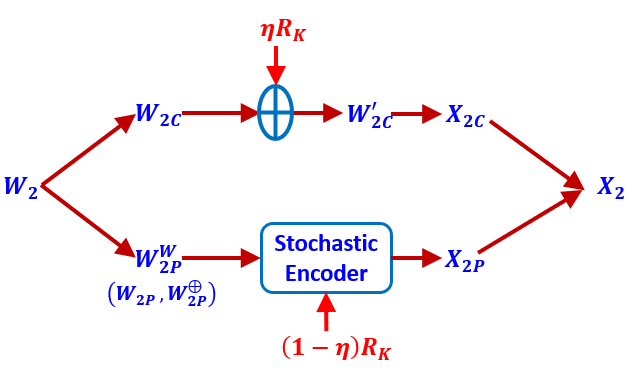}
		\caption{The encoding process at transmitter $2$.}
		\label{fig:encoding process}
	\end{center}
\end{figure}

\textbf{\emph{Decoding and error probability analysis:}} 
Receiver~$1$ decodes its intended message and common message of transmitter~$2$ jointly. The private message of transmitter~$2$ is treated as noise while decoding other messages. Consider the following event where the message of transmitter~$1$, the common confidential message of transmitter~$2$ and output at receiver~$1$ are jointly typical:
\begin{align}
E_{ij}=\{(X_1^n(i),X_{2C}^n(j),Y_1^n) \in A_{\epsilon}^{(n)}\},
\end{align}
Due to the symmetry in the codebook construction process, without loss of generality, one can assume that $(i,j) = (1,1)$ was sent. Then, the probability of error can be written as follows.
\begin{align}
    &P_{e1}^{(n)}=P(E_{11}^c\bigcup\cup_{(i,j)\neq(1,1)}E_{i,j}), \label{eq:error2}
\end{align}
Following a similar analysis as  in \cite{cover2006elements}, one can show that the probability of error at receiver~$1$, i.~e., $P_{e1}^{(n)}$ tends to zero if the following conditions are satisfied
\begin{align}
R_1&\leq I(X_1;Y_1|X_{2C}),\nonumber\\
R_{2C}&\leq I(X_{2C};Y_1|X_1),\nonumber\\
R_1+R_{2C}&\leq I(X_1,X_{2C};Y_1), \label{eq:error3}
\end{align}
The decoding at receiver~$2$ happens at two stages as discussed below:
\begin{enumerate} 
\item The decoder decides in favour of $(\hat{W}_{2C})$ has been sent if $(X_{2C}^n(\hat{W}_{2C}), Y_2^n)$ is jointly typical; otherwise an error is declared.
\item If such  $(\hat{W}_{2C})$ is found, the decoder needs to determine an $\hat{l}$, such that $(X_{2p}^n(\hat{l}), Y_2^n)$ is jointly typical given the knowledge of $X_{2C}^n(\hat{W}_{2C})$, where $\hat{l} \in [1:2^{n(R_{2p} + R_{2p}' + R_{2p}^{\oplus}}]$. If no such unique $\hat{l}$ is found, then an error is declared. The decoder decides in favour of $\hat{W}_{2p}$ the index of the bin $X_{2p}^n(\hat{l})$ belongs to. To determine $W_{2p}^{\oplus}$, the decoder needs to determine the index of $X_{2p}^n(\hat{l})$ to which section the codeword belongs $(\hat{l'})$. As the receiver has access to the key, it can determine $\hat{W}_{2p}^{\oplus} = \hat{l'} \oplus R_K^{\oplus}$, where $\oplus $ is the modulo addition over $[1: 2^{(1 - \eta) R_K}]$. 
\end{enumerate}
For driving the error probability to zero in decoding the common message at receiver~$2$, following condition is required to be satisfied.
\begin{align}
	R_{2C} \leq I(X_{2C}; Y_2) \label{eq:ach1}
\end{align}
Due to the secrecy constraint, the common message also needs to satisfy the following condition
\begin{align}
	R_{2C} \leq \eta R_K. \label{eq:ach2}
\end{align}
After decoding the common message, receiver~$2$ attempts to decode the private message $W_{2P}$ given the knowledge of the common message. To drive the error probability to zero, following condition is required to be satisfied
\begin{align}
	R_{2P} + R_{2P}' \leq I(X_{2P}; Y_2|X_{2C}), \label{eq:ach3}
\end{align}
For decoding the message $W_{2P}^{\oplus}$, following conditions are required to be satisfied
\begin{align}
& R_{2P}^{\oplus} \leq (1 - \eta) R_K, \nonumber \\
& R_{2P} + R_{2P}^{\oplus} \leq I(X_{2P}; Y_2|X_{2C}). \label{eq:ach4}
\end{align}
Using the relation $R_2 = R_{2C} + {R_{2P}^W} = R_{2C} + R_{2P} + R_{2P}^\oplus$ and Fourier-Motzkin Elimination \cite{el2011network}, one can obtain the achievable rate region stated in the theorem. In the following, the amount of rate that needs to be sacrificed $( R_{2P}')$ to confuse the unintended receiver is determined for guaranteeing secrecy of the private message of transmitter $2$.\\
 \textbf{\emph{Secrecy analysis:}}\\
 Consider the following
 \begin{align}
     I(W_2;Y_1^n)&\stackrel{(a)} = I(W_{2C}, W_{2P}, W_{2P}^\oplus; Y_1^n),\nonumber\\
     &=I(W_{2P} ; Y_1^n) + I(W_{2P}^\oplus; Y_1^n|W_{2P})\nonumber\\
     & \qquad + I(W_{2C}; Y_1^n|W_{2P},{W_{2P}^\oplus}),\nonumber\\
     &\stackrel{(b)} = I(W_{2P}; Y_1^n), \nonumber\\
     &\stackrel{(c)} \leq H(W_{2P}) - H(W_{2P}|Y_1^n, X_1^n),\nonumber\\
     &= I(W_{2P}; Y_1^n|X_1^n),\nonumber\\
     &\stackrel{(d)} = I(W_{2P}; S_1^n), \label{eq:sec1}
 \end{align}
 where (a) is due to the splitting of the message of transmitter~$2$ into three parts; (b) the message $W_{2C}$ is independent of the output due to the one-time padding scheme and from the code construction it can be seen that the codeword corresponding to the common confidential message is independent of the private message. It can also be observed that $W_{2P}^{\oplus}$ is independent of $Y_1^n$ due to the non-availability of the key at receiver~$1$ and one-time padding; (c) is obtained from the definition of mutual information and the fact that conditioning cannot increases the entropy and (d) is obtained using the independence between $W_{2P}$ and $X_1^n$, where $S_1^n  \triangleq h_{21}X_{2P}^n + Z_1^n$. From the above, it is observed that transmitter~$2$ forms a hypothetical Gaussian wiretap channel with receiver $2$ (as legitimate receiver) and receiver $1$ (as eavesdropper). Using the approach in~\cite{bloch2013strong}, one can ensure that $I(W_{2P} ; S_1^n) \rightarrow 0$ as $n\rightarrow \infty$ if the condition in (\ref{equivocation computation}) holds true. 
\end{proof}
\textit{Remark:} The above result holds for the scenario when the channel strength of the interfering link is weaker than the direct link, i.~e., $h_{22} \leq h_{21}$. For guaranteeing security of the private message  using wiretap coding, this condition is required to be satisfied. However, when transmitter~$1$ sends artificial noise drawn from Gaussian distribution along with its message, it is possible to send private message using stochastic encoding. The achievable result when transmitter~$1$ sends artificial noise along with its message is stated in the following corollary.

\begin{corollary} \label{corollary-key splitting}
The following rate region is achievable with artificial noise (AN) transmission by transmitter $1$ using the result in Theorem~\ref{theorem-key splitting}:
\begin{align}
R_s=conv  \bigcup_{0\leq (\lambda_1,\lambda_2,\beta_1,\beta_2,\eta)\leq 1} R_s^{ZIC}(\lambda_1,\lambda_2,\beta_1,\beta_2,\eta),
\end{align}
where conv represents the convex hull of the rate region, and
\begin{align}
&R_s^{ZIC}(\lambda_1,\lambda_2,\beta_1,\beta_2,\eta)=\{(R_1,R_2):R_1\geq 0, R_2\geq 0 ,\nonumber\\
&R_1\leq 0.5 \log \biggl(1+\frac{h_{11}^2P_{1M}} {1+h_{11}^2P_{1A}+h_{21}^2P_{2P}}\biggr),\nonumber\\ 
&R_2\leq \min\biggl\{0.5\log \biggl(1+\frac{h_{21}^2P_{2C}} {1+h_{11}^2P_{1A}+h_{21}^2P_{2P}}\biggr),\nonumber\\
&\qquad0.5\log\biggl(1+\frac{h_{22}^2P_{2C}}{1+h_{22}^2P_{2P}}\biggr),\eta R_K\biggr\}\nonumber\\
&\qquad+\min\{0.5\log (1+ h_{22}^2P_{2P}),0.5\log (1+h_{22}^2P_{2P})\nonumber\\
&\qquad-R_{2P}'+(1-\eta)R_K\},\nonumber \\
&R_1+R_2\leq 0.5\log\biggl (1+\frac{h_{11}^2 P_{1M}+h_{21}^2P_{2C}}{1+h_{11}^2P_{1A}+h_{21}^2P_{2P}}\biggr)\nonumber\\
&\qquad\qquad+\min\{0.5\log(1+h_{22}^2P_{2P}),0.5\log(1+h_{22}^2P_{2P})\nonumber\\
&\qquad\qquad-R_{2P}'+(1-\eta)R_K\}, 
\end{align}
where $P_{1M}=\lambda_1\beta_1P_1$, $P_{1A}=(1-\lambda_1)\beta_1P_1$, $P_{2P}=\lambda_2 \beta_2 P_2$,  $P_{2C}=(1-\lambda_2)\beta_2P_2$ and $R_{2P}'=0.5\log (1+\frac{h_{21}^2P_{2P}}{1+h_{11}^2P_{1A}})$.
\end{corollary}
\textit{Remarks}: 
\begin{itemize}
\item Transmitter~$1$ sends its message along with artificial noise which is drawn from the Gaussian distribution $\mathcal{N}(0, P_{1A})$.  In corollary~\ref{corollary-key splitting}, the parameters $\beta_i$ and $\lambda_i $ $(i=1,2)$ are the power control  and power splitting parameters of transmitter $1$ and $2$, respectively. 
\item The term $\eta$ is the key rate splitting parameter. Depending on whether $\eta$ is $1$ or $0$, the achievable scheme correspond to two distinct schemes. When $\eta = 1$, the entire key rate is used for the common confidential message. On the other hand, when $\eta = 0$, the entire key is used as a part of the wiretap coding and in this case, no common confidential message is sent by transmitter~$2$. The results corresponding to these two cases are stated in the following corollaries.
\end{itemize}
\begin{corollary} \label{corollary-rate splitting approach}
When $\eta = 1$, the achievable rate region stated in Corollary~\ref{corollary-key splitting} simplifies to the following
\begin{align}
   R_s=conv  \bigcup_{0\leq (\lambda_1,\lambda_2,\beta_1,\beta_2)\leq 1,\eta=1} R_s^{ZIC}(\lambda_1,\lambda_2,\beta_1,\beta_2,\eta),
\end{align}
where conv represents the convex hull of the rate region, and
\begin{align}
&R_s^{ZIC}(\lambda_1,\lambda_2,\beta_1,\beta_2,\eta)=\{(R_1,R_2):R_1\geq 0, R_2\geq 0 ,\nonumber\\
&R_1\leq 0.5 \log \biggl(1+\frac{h_{11}^2P_{1M}} {1+h_{11}^2P_{1A}+h_{21}^2P_{2P}}\biggr),\nonumber\\ 
&R_2\leq \min\biggl\{0.5\log \biggl(1+\frac{h_{21}^2P_{2C}} {1+h_{11}^2P_{1A}+h_{21}^2P_{2P}}\biggr),\nonumber\\
&\qquad0.5\log\biggl(1+\frac{h_{22}^2P_{2C}}{1+h_{22}^2P_{2P}}\biggr), R_K\biggr\}\nonumber\\
&\qquad+\min\{0.5\log (1+ h_{22}^2P_{2P}),0.5\log (1+h_{22}^2P_{2P})\nonumber\\
&\qquad-R_{2P}'\},\nonumber \\
&R_1+R_2\leq 0.5\log\biggl (1+\frac{h_{11}^2 P_{1M}+h_{21}^2P_{2C}}{1+h_{11}^2P_{1A}+h_{21}^2P_{2P}}\biggr)\nonumber\\
&\qquad\qquad+\min\{0.5\log(1+h_{22}^2P_{2P}),0.5\log(1+h_{22}^2P_{2P})\nonumber\\
&\qquad\qquad-R_{2P}'\},  
\end{align}

\textit{Remark:} In this case, the transmitter~$2$ sends common confidential message and private message ($W_{2P}$) and the entire key rate is utilized for sending the common confidential message. 
\end{corollary}
\begin{corollary}\label{corollary-key is used as a part of wiretap coding}
When $\eta = 0$ and $\lambda = 1$, the achievable rate region stated in Corollary~\ref{corollary-key splitting} simplifies to the following
\begin{align}
    R_s=conv  \bigcup_{0\leq (\beta_1,\beta_2)\leq 1} R_s^{ZIC}(\beta_1,\beta_2),
\end{align}
where conv represents the convex hull of the rate region, and
\begin{align}
&R_s^{ZIC}(\beta_1,\beta_2)=\{(R_1,R_2):R_1\geq 0, R_2\geq 0 ,\nonumber\\
&R_1\leq 0.5 \log \biggl(1+\frac{h_{11}^2P_1}{1+h_{21}^2P_2}\biggr),\nonumber\\ 
&R_2\leq\min\{0.5\log (1+ h_{22}^2P_2),0.5\log (1+h_{22}^2P_2)\nonumber\\
&\qquad-0.5\log (1+h_{21}^2P_2) + R_K\},
\end{align}
\end{corollary}
\textit{Remarks:}
\begin{itemize}
	\item In this case, the entire key rate is used as a part of the wiretap coding and this can help to reduce the loss in rate due to stochastic encoding. In this case, receiver~$1$ decodes its intended message by treating the message of transmitter~$2$ as noise. 
	\item By setting $P_1 =0$, one can obtain the result for the wiretap channel with secret key \cite{ardestanizadeh2009wiretap}. 
\end{itemize}
 \subsection{Key is  used as One Time Pad}
 In this case, the key is used to protect the message of the transmitter~$2$ and receiver~$1$ decodes its message by treating interference as noise. This scheme has less decoding complexity in comparison to other schemes. It can be considered as the secure version of treating interference as noise scheme. The result is stated in the following theorem.
 \begin{theorem}\label{theorem-key as a one time pad}
  The achievable rate region of the 2-user Gaussian ZIC when the key is used as a one-time pad is given by
  \begin{align}
R_s=conv  \bigcup_{0\leq (\beta_1,\beta_2)\leq 1} R_s^{ZIC}(\beta_1,\beta_2),
\end{align}
where  conv represents the convex hull of the rate region, and 
\begin{align}
&R_s^{ZIC}(\beta_1,\beta_2)=\{(R_1,R_2):R_1\geq 0, R_2\geq 0 ,\nonumber\\
&R_1\leq 0.5 \log \biggl(1+\frac{h_{11}^2\beta_1P_1} {1+h_{21}^2\beta_2P_2}\biggr),\nonumber\\
&R_2\leq \min\{R_K,0.5 \log (1+h_{22}^2\beta_2 P_2)\}\}.\label{key as a one time pad}
\end{align}
 \end{theorem}
\begin{proof}\label{proof-achievable rate region of one time pad}
In this case, the secret key $K$ is used to encrypt the message of transmitter~$2$, i.~e., $W_{2}' = W_{2} \oplus K$, where $W_{2}'$ is the encrypted message. Then, the encrypted  message is encoded using the codebook  $\mathcal{C}_{2}$ which contains $2^{nR_2}$ i.i.d. sequences of length $n$. Each entries of the codebook are drawn randomly from $\mathcal{N}(0, \beta_2 P_2)$. In a similar way, the codebook for transmitter~$1$ is generated from $\mathcal{N}(0, \beta_1 P_1)$. Receiver~$1$ decodes its intended message using treating interference as noise. The rate of the message at transmitter~$2$ is limited by the key rate and its reliability constraint at receiver~$2$. As the message $W_2$ is independent of $X_2$ due to one-time pad, the message $W_2$ is independent of the output at receiver~$1$. Hence, the message of the transmitter~$2$ remains secure. As the receiver~$2$ knows the key, it can decode the message $W_2$ provided it can decode the codeword $X_2^n$.
\end{proof}
\section{Outer Bounds}\label{sec:OB}
In this section, outer bounds on the secrecy capacity region of the 2-user Gaussian Z-IC with shared key are presented. The result in Theorem~\ref{Outer bound for weak/moderate interference regime} provides an outer bound on the sum rate and this outer bound is applicable when the strength of interference is less than that of the signal strength, i.~e., $\text{INR} \leq \text{SNR}$. Derivation of the outer bound needs to take account of the secrecy constraint at receiver~$1$ and shared key between transmitter~$2$ and receiver~$2$.  The result stated in Theorem~\ref{Outer bound for high interference regime} provides an outer bound on the rate of user~$2$. The main novelty of the proofs lies in the choice of side-information provided to the receiver and using the secrecy constraint at the receiver~$1$. Using the above bounds along with the conventional bound on the capacity of the point-to-point channel, one can obtain bound on the secrecy capacity region. The outer bound on the sum rate for the weak/moderate interference regime is stated below.
\begin{theorem}\label{Outer bound for weak/moderate interference regime}
The sum rate of the $2$-user Gaussian Z-interference channel with secret key of rate $R_K$ is upper bound by 
\begin{align}
  R_{sum} & \leq\log(1+SNR) - 0.5\log(1 + INR) + R_K, 
  \end{align}
where $SNR=h_d^2P$, $INR=h_c^2P$ and $ SNR > INR$. 
\end{theorem}
\begin{proof}
Using Fano's inequality, the sum rate is upper bounded as
\begin{align}
&n(R_1+R_2)\nonumber\\
    &\leq I(W_1;Y_1^n) + I(W_2; Y_2^n, K_2) + n\epsilon_n,\nonumber\\
    &\stackrel{(a)}\leq I(W_1; Y_1^n) + I(W_2; Y_2^n, K_2) - I(W_2;Y_1^n) + n\epsilon_n,\nonumber\\
    &\stackrel{(b)}=I(W_1;Y_1^n) + I(W_2;Y_2^n) + I(W_2;K_2|Y_2^n)\nonumber\\
    &\qquad-I(W_2;Y_1^n,S_2^n) + I(W_2;S_2^n|Y_1^n) + n\epsilon_n, \label{weakOB eq1} 
\end{align}
where (a) is obtained using the secrecy condition at receiver $1$ and (b) is obtained by using the chain rule for mutual information $I(W_2;Y_1^n,S_2^n) = I(W_2;Y_1^n) + I(W_2; S_2^n|;Y_1^n) $, where $S_2^n = h_c X_2^n + Z_1^n$.

Consider the following terms from the above equation
\begin{align}
&I(W_1; Y_1^n) + I(W_2; S_2^n|Y_1^n)\nonumber\\
 &\leq I(W_1,X_1^n;Y_1^n) + I(W_2;S_2^n|Y_1^n),\nonumber\\
 &\stackrel{(a)} = I(X_1^n;Y_1^n) + h(S_2^n|Y_1^n) - h(S_2^n|Y_1^n, W_2), \nonumber\\
    &\stackrel{(b)} = I(X_1^n;Y_1^n) + h(S_2^n,Y_1^n) - h(Y_1^n) \nonumber\\
    &\qquad-h(S_2^n|Y_1^n, W_2), \nonumber\\
    &=I(X_1^n;Y_1^n)+h(S_2^n)+h(Y_1^n|S_2^n)-h(Y_1^n) \nonumber\\
    &\qquad-h(S_2^n|Y_1^n, W_2), \nonumber\\
    &\stackrel{(c)}{=} I(X_1^n;Y_1^n)+h(Y_1^n|X_1^n)+h(Y_1^n|S_2^n)-h(Y_1^n) \nonumber\\
    &\qquad-h(S_2^n|Y_1^n,W_2), \nonumber\\
    & = h(Y_1^n|S_2^n) - h(S_2^n|Y_1^n,W_2), \nonumber\\
    &\leq h(Y_1^n|S_2^n) - h(S_2^n|Y_1^n, W_2, K_2, X_2^n), \nonumber\\
    &\stackrel{(d)}{=} h(Y_1^n|S_2^n) - h(S_2^n|Y_1^n,X_2^n),\nonumber\\
    &= h(Y_1^n|S_2^n)-h(S_2^n, Y_1^n|X_2^n) + h(Y_1^n|X_2^n),\nonumber\\
    &= h(Y_1^n|S_2^n) - h(S_2^n|X_2^n) - h(Y_1^n|X_2^n, S_2^n) + h(Y_1^n|X_2^n),\nonumber\\
    &\stackrel{(e)}{=} h(S_1^n) - h(Z_1^n), \qquad \text{where}\quad S_1^n = h_d X_1^n + Z_1^n, \label{weakOB eq2}
\end{align}
where (a) is obtained using the fact that $W_1 \rightarrow X_1^n \rightarrow Y_1^n$; (b) is obtained from the chain rule of entropy; (c) is obtained using the relation $h(Y_1^n|X_1^n) = h(h_c X_2^n+Z_1^n) = h(S_2^n)$; (d) is obtained using the fact that $(W_2, K_2) \rightarrow X_2^n \rightarrow S_2^n$ and (e) is obtained by observing that the first and third term cancel with each other, where $S_1^n=h_d X_1^n+Z_1^n$. 

Consider the following terms in (\ref{weakOB eq1})
\begin{align}
    &I(W_2;Y_2^n) + I(W_2;K_2|Y_2^n) - I(W_2;Y_1^n, S_2^n),\nonumber\\
    &\leq I(W_2;Y_2^n) - I(W_2;S_2^n) + I(W_2;K_2|Y_2^n), \nonumber\\
    &\stackrel{(a)}\leq H(W_2|S_2^n) - H(W_2|Y_2^n,S_2^n) + I(W_2;K_2|Y_2^n),\nonumber\\
    &= I(W_2;Y_2^n|S_2^n) + I(W_2;K_2|Y_2^n),\nonumber\\
    &\leq I(W_2,K_2;Y_2^n|S_2^n) + I(W_2;K_2|Y_2^n),\nonumber\\
    &\leq I(X_2^n;Y_2^n|S_2^n) + H(K_2),\nonumber\\
    &\stackrel{(b)}\leq n[I(X_2;Y_2) - I(X_2;S_2)] + nR_K,\label{weakOB eq3}
\end{align}
where (a) is obtained from the definition of mutual information and the fact that conditioning can not increase the entropy; (b) is obtained from the degradedness condition of the wiretap channel. The quantities $X_2$, $Y_2$ ans $S_2$ form a hypothetical wiretap channel where $S_2$ is the output received at the eavesdropper.

Finally, using (\ref{weakOB eq2}) and (\ref{weakOB eq3}) in (\ref{weakOB eq1}), the sum rate becomes
\begin{align}
R_{sum}&\leq \log(1+h_d^2P) - \frac{1}{2}\log(1+h_c^2P) + R_K.
\end{align}
\end{proof}
\textit{Remarks}
\begin{enumerate}
\item The outer bound in Theorem \ref{Outer bound for weak/moderate interference regime} can be extended to the asymmetric scenario and stated below. For the following outer bound to hold, it should satisfy the following condition: $SNR_2 > INR_1$.
\begin{align}
    R_2 & \leq 0.5\log(1+SNR_2) - 0.5\log(1+INR_1) + R_K,
\end{align}
where $SNR_1=h_{11}^2P_1$, $SNR_2=h_{22}^2P_2$ and $INR_1=h_{21}^2P_2$.
\item When $R_K=0$, the system reduces to the $2$-user Gaussian Z-Interference channel without secret key. Therefore, the outer bound in Theorem \ref{Outer bound for weak/moderate interference regime} becomes $\log(1+h_d^2P)-0.5\log(1+h_c^nP)$, which is studied in \cite{li2008secrecy}.
\end{enumerate}
\begin{theorem}\label{Outer bound for high interference regime}
The secrecy rate of user~$2$  in case of 2-user Gaussian Z-IC with shared secret key of rate $R_K$ is upper bounded as follows
\begin{align}
   R_2 \leq 0.5\log\lb 1 + SNR - \frac{SNR.INR}{1 + SNR + INR} \rb + R_K. 
\end{align}
\end{theorem}
\begin{proof}
Using Fano's inequality, the rate of user $2$ is upper bounded as
\begin{align}
nR_2&\leq I(W_2;Y_2^n,K_2) + n\epsilon_n,\nonumber\\
    &\stackrel{(a)}\leq I(W_2;Y_2^n,Y_1^n, K_2) + n\epsilon_n,\nonumber\\
    &\stackrel{(b)}\leq I(W_2;Y_1^n) + I(W_2; Y_2^n|Y_1^n)\nonumber\\
    &\quad + I(W_2;K_2|Y_1^n,Y_2^n) - I(W_2;Y_1^n) + n\epsilon_n, \nonumber\\
    & = I(W_2;Y_2^n|Y_1^n) + I(W_2;K_2|Y_1^n, Y_2^n) + n\epsilon_n, \nonumber\\
    &\leq h(Y_2^n|Y_1^n) - h(Y_2^n|W_2, Y_1^n) + nR_K + n\epsilon_n,\nonumber\\
R_2&\stackrel{(b)}\leq 0.5\log \textstyle \sum_{Y_2|Y_1}+R_K,\\
\text{where}\quad&\textstyle \sum_{Y_2|Y_1}=E[Y_2^2]-\frac{ E[Y_2Y_1]^2}{E[Y_1^2]},\nonumber\\
&\qquad\quad=1 + SNR - \frac{SNR.INR}{1 + SNR + INR}.\nonumber
\end{align}
where (a) is obtained using the fact that providing additional side-information cannot reduce the rate; (b) is obtained using the secrecy condition at receiver~$1$ and (c) is obtained from the fact that the Gaussian distribution maximizes the differential entropy for a given power constraint.

Therefore, the rate of user $2$ is finally upper bounded as follows
\begin{align}
&R_2 \leq 0.5\log\lb 1 + SNR - \frac{SNR.INR}{1 + SNR + INR} \rb + R_K. 
\end{align}
\end{proof}
\textit{Remarks}
\begin{enumerate}
\item The outer bound in Theorem~\ref{Outer bound for high interference regime} can be extended to the asymmetric scenario and is stated in the following
\begin{align}
    R_2 &\leq  0.5\log\biggl(1 + SNR_2 - \frac{SNR_2.INR_1}{1 + SNR_1 + INR_1}\biggr) \nonumber\\
   &\qquad\qquad\qquad\qquad\qquad\qquad + R_K.
\end{align}
\item When the system is operating under interference limited regime and $INR >> SNR$, the rate in Theorem~\ref{Outer bound for high interference regime} is simplified as follows:
\begin{align}
	R_2 & \leq  0.5\log\lb 1 + SNR - \frac{SNR.INR}{1 + SNR + INR} \rb \nonumber\\
	&\qquad\qquad\qquad\qquad\qquad\qquad\qquad\qquad\quad+ R_K, \nonumber \\
	& \approx 0.5\log\lb 1 + SNR -  \frac{SNR.INR}{SNR + INR} \rb + R_K, \nonumber \\
	& \approx R_K
\end{align}
where the last equation is obtained using the fact that $INR >> SNR $ and $ SNR + INR = INR$. Hence, in this case, rate of user~$2$ is primarily limited by the rate of the key. 
\end{enumerate}
\section{Secure Generalized Degrees of Freedom Region (SGDOF)}
 The primary objective of this section is to study how the different schemes behave when the key rate is scaled with respect to underlying channel conditions such as SNR/INR. This can give new insights on how one needs to scale the key rate in interference limited scenario. To the best of author's knowledge, this kind of scaling of key rate with respect to underlying channel characteristics such as SNR/INR has not been explored in this existing literature. To capture such behaviour, the notion of secure generalized degree of freedom (SGDOF) is used in the paper. The notion of GDOF was proposed in~\cite{etkin-TIT-2008} which is a good approximation of capacity region under high SNR and INR. The characterization of the capacity region of the multiuser network has remained as a challenging problem even without secrecy constraint. Instead of the exact characterization of the capacity region, approximate characterization of it has received significant attention in the last decade.  In this paper, a modified definition of a secure GDOF region is proposed to take account of the key rate shared between the users. To simplify the presentation, it is assumed that $h_{11} = h_{22} = h_d$ and $h_{21} = h_c$. The signal-to-noise ratio and interference-to-noise ratio are defined as $SNR \triangleq h_d^2 P$ and $INR \triangleq h_c^2 P$. The secure GDOF region for the 2-user Gaussian ZIC with shared key is defined as follows. 
\begin{align}
D(\alpha, \gamma) \triangleq \displaystyle\lim_{(SNR, INR) \to \infty,  \alpha \text{ and } \gamma \text{ are fixed }} \tilde{D}(SNR, INR, R_K) \label{eq:gdof1}
\end{align}
where $\alpha \triangleq \frac{\log INR}{\log SNR}$,  $\gamma \triangleq \frac{R_K}{0.5\log SNR}$, and $\tilde{D}(SNR, INR, R_K)$ is defined as follows
\begin{align}
 \tilde{D}(SNR, INR, R_K) & =\biggl\{\lb \frac{R_1}{0.5\log SNR}, \frac{R_2}{0.5\log SNR}\rb \nonumber \\
&: (R_1,R_2) \in \mathcal{C} (SNR, INR, R_K) \Biggr \}. \nonumber
\end{align} 
and $\mathcal{C} (SNR, INR, R_K)$ denotes the secrecy capacity region of the considered model. In the following, the secure GDOF region for different schemes is characterized. 
\subsection{SGDOF region: Key rate Splitting}
To characterize the GDOF region of the achievable scheme in Section~\ref{achievable scheme:key splitting}, the following power allocation is used for the private message and common confidential message
\begin{align}
P_{2P} = \min\lb P, \frac{1}{h_c^2} \rb = \frac{1}{h_c^2} \text{ and } P_{2C} = \lb P - \frac{1}{h_c^2} \rb^{+} \label{gdof2}
\end{align}
The power of the private message is chosen such that it causes almost the same damage as caused by the noise at the receiver. This power allocation also ensures that the loss in rate due to stochastic encoding ($R_{2P}' = 0.5$) does not scale with INR. 
\begin{theorem}\label{gdof-key splitting} 
The SGDOF region of the scheme in Corollary~\ref{corollary-key splitting} without artificial noise transmission is given by
\begin{align}
      d_1(\alpha,\gamma)&\leq 1, \nonumber\\
      d_2(\alpha,\gamma)&\leq  \min\{\alpha, \eta\gamma\} + 1 - \alpha, \nonumber \\
      d_1(\alpha,\gamma) + d_2(\alpha,\gamma)&\leq2 -\alpha.
 \end{align}
 \end{theorem}
 \begin{proof}
 Using the power allocation in (\ref{gdof2}) in Corollary~\ref{corollary-key splitting}, the rate of user~$1$ becomes
\begin{align}
 R_1  & \leq 0.5\log\lb 1+\frac{h_d^2P}{2} \rb,\label{eq:gdofscheme1}\\
& =0.5\log SNR-0.5+\mathcal{O}(1).\nonumber
\end{align}
In this case, it is assumed that no artificial noise is sent by the transmitter~1.  Hence, the achievable SGDOF of user $1$ becomes,
\begin{align}
   & d_1(\alpha,\gamma) \leq 1. \label{eq:gdofscheme1a}
\end{align}
The rate of user $2$ reduces to,
\begin{align}
R_2& \leq \min\{0.5\log INR, \eta R_K\}\nonumber\\
    &+\min\{(0.5\log SNR-0.5\log INR),(0.5\log SNR \nonumber\\
    &\qquad-0.5\log INR)+(1-\eta)R_K\} + \mathcal{O}(1). \label{eq:gdofscheme2}
\end{align}
Therefore, the achievable SGDOF for user $2$ becomes,
\begin{align}
   d_2(\alpha,\gamma) & \leq \min\{\alpha,\eta\gamma\} + \min\{1-\alpha,1-\alpha+(1-\eta)\gamma\}, \nonumber \\
  & = \min\{\alpha,\eta\gamma\} + 1 - \alpha. \label{eq:gdofscheme2a}
\end{align}
The achievable sum-rate under this power allocation simplifies to the following
\begin{align}
    R_1+R_2 &\leq 0.5\log(1+SNR+INR)+\min\{0.5\log SNR\nonumber\\
    &\qquad-0.5\log INR,0.5\log SNR-0.5\log INR \nonumber\\
    &\qquad+(1-\eta)R_K\}+\mathcal{O}(1). \label{eq:gdofscheme3}
\end{align}
Hence, the achievable sum SGDOF becomes,
\begin{align}
d_1(\alpha,\gamma) + d_2(\alpha,\gamma) & \leq 1 + \min\{1-\alpha,1-\alpha + (1-\eta)\gamma\}, \nonumber \\
& = 2 - \alpha.  \label{eq:gdofscheme3a}
\end{align}
Using (\ref{eq:gdofscheme1a}), (\ref{eq:gdofscheme2a}) and (\ref{eq:gdofscheme3a}), the result in the Theorem can be obtained.
 \end{proof}
Recall that by setting $\eta = 1$, one can obtain the achievable rate region corresponding to the scenario where the entire key rate is used for encoding the common confidential message and when $\eta = 0$, the entire key rate is used as a part of the stochastic encoding. The SGDOF when $\eta =1$ is given by the following corollary.
 \begin{corollary}\label{gdof-rate splitting}
By setting $\eta = 1$ in Theorem~\ref{gdof-key splitting}, the SGDOF region corresponding to the achievable scheme when the entire key is used for encoding the common confidential message is given below
\begin{align}
    d_1(\alpha,\gamma)&\leq 1,\nonumber\\
    d_2(\alpha,\gamma)&\leq \min(\alpha,\gamma)+1-\alpha,\nonumber\\
    d_1(\alpha,\gamma)+d_2(\alpha,\gamma)&\leq 2-\alpha.
\end{align}
 \end{corollary}
 In the following, the achievable SGDOF region for $\eta=0$ (when key is used as a part of wiretap coding) is presented. In this case, two SGDOF regions are obtained with different power allocations. Using time-sharing between the two SGDOF regions  $\mathcal{D}^{I}(\alpha, \gamma)$ and $\mathcal{D}^{II}(\alpha, \gamma)$, the SGDOF region is obtained. 

The SGDOF region when $\eta=0$ is given by the following corollary.
 \begin{corollary}\label{gdof-key is used as a part of WC}
 The SGDOF region is obtained by taking time-sharing between the two SGDOF regions as given below where the key is used as a part of wiretap coding
 \begin{equation}
\left. \begin{array}{l}
d_{1}(\alpha,\gamma) \leq 1-\alpha, \\
d_{2}(\alpha,\gamma)  \leq  \min(1,1-\alpha+\gamma).
\end{array} \rcb \mathcal{D}^{I}(\alpha, \gamma) \label{fic-fixed power split} 
\end{equation}
\begin{equation}
\left. \begin{array}{l}
d_{1}(\alpha,\gamma)\leq 1, \\ 
d_{2}(\alpha,\gamma)\leq 1 - \alpha.\label{fic-time sharing}
\end{array} \rcb \mathcal{D}^{II}(\alpha, \gamma) 
\end{equation}
\begin{proof}
The SGDOF region in (\ref{fic-fixed power split}) is derived by considering the following power allocation $P_1 = P_2 = P$. The rate of user~$1$ and $2$ reduces to the following form
\begin{align}
R_1&\leq 0.5\log\biggl(1+\frac{h_d^2P}{1+h_c^2P}\biggr), \nonumber \\
&=0.5\log(1+SNR+INR)-0.5\log(1+INR),\nonumber\\
&=0.5\log SNR-0.5\log INR+\mathcal{O}(1), \label{eq:godfachscheme5aa}
\end{align}
hence, the achievable SGDOF for user $1$ becomes
\begin{align}
    d_1(\alpha,\gamma)&\leq 1-\alpha. \label{eq:godfachscheme5a}
\end{align}
Similarly, the rate of user $2$ is simplified to,
\begin{align}
R_2& \leq  \min\{0.5\log (1+h_d^2P), 0.5\log (1+h_d^2P)\nonumber\\
&\qquad\qquad-0.5\log (1+h_c^2P)  + R_K\},\\
& = \min\{0.5\log SNR, 0.5\log SNR-0.5\log INR  + R_K\}  \nonumber\\
& \qquad \qquad \qquad\qquad + \mathcal{O}(1). \label{eq:godfachscheme10}
\end{align}
Hence, the achievable SGDOF for user $2$ becomes
\begin{align}
    d_2(\alpha,\gamma) & \leq \min\{1, 1-\alpha+\gamma\}. \label{eq:godfachscheme11}
\end{align}
Using (\ref{eq:godfachscheme5a}) and (\ref{eq:godfachscheme11}), the SGDOF region $\mathcal{D}^{I}(\alpha, \gamma)$ in the theorem is obtained.

For the other region, the following power allocation is considered $P_1=P$, and $P_2=\frac{1}{h_c^2}$. In this case, no common confidential message is sent and $\eta = 0$. Using this, in Theorem~\ref{gdof-key splitting}, the following GDOF region is obtained:
\begin{align}
	d_1(\alpha,\gamma)&\leq 1, \nonumber\\
	d_2(\alpha,\gamma)&\leq  1 - \alpha, \nonumber \\
	d_1(\alpha,\gamma) + d_2(\alpha,\gamma)&\leq2 -\alpha.
\end{align}
Given the first two equations, the last equation is redundant. Hence, the SGDOF region $\mathcal{D}^{II}(\alpha, \gamma)$ in the theorem is obtained.
\end{proof}
 \end{corollary}
 \subsection{SGDOF region: Key is Used as a One Time Pad}
The following theorem gives the SGDOF region when a key is used as a one-time pad. In this case, two SGDOF regions are obtained with different power allocations. Using time-sharing between the two SGDOF regions $\mathcal{D}^{I}(\alpha, \gamma)$ and $\mathcal{D}^{II}(\alpha, \gamma)$, the SGDOF region is obtained. 
\begin{theorem}\label{theorem-gdof of key as a one time pad}
The SGDOF region is obtained by taking time-sharing between the two SGDOF regions as given below when the key is used as a one-time pad
\begin{equation}
\left. \begin{array}{l}
	 d_{1}(\alpha,\gamma) \leq 1-\alpha, \\
	d_{2}(\alpha,\gamma)  \leq \min(\gamma, 1)
\end{array} \rcb \mathcal{D}^{I}(\alpha, \gamma) \label{eq:otp-fixed-power-split} 
\end{equation}
\begin{equation}
\left. \begin{array}{l}
d_{1}(\alpha,\gamma) \leq 1, \\
d_{2}(\alpha,\gamma) \leq \min(\gamma,1-\alpha).
\end{array} \rcb \mathcal{D}^{II}(\alpha, \gamma) \label{eq:otp-time-sharing} 
\end{equation}
\end{theorem}
\begin{proof}\label{proof:gdof-key as a one time pad}
The above result is obtained using the result from Theorem~\ref{theorem-key as a one time pad}, where the following power allocation is considered: $P_1 = P_2 = P$. The rate for user~$1$ becomes
\begin{align}
R_1&\leq 0.5\log\biggl(1+\frac{h_d^2P}{1+h_c^2P}\biggr),\\
&=0.5\log(1+SNR+INR)-0.5\log(1+INR),\nonumber\\
&=0.5\log SNR-0.5\log INR+\mathcal{O}(1). \label{eq:godfachscheme5}
\end{align}
Hence, the achievable SGDOF of user $1$ becomes
\begin{align}
    d_1(\alpha,\gamma)&\leq 1-\alpha. \label{eq:godfachscheme6a}
\end{align}
Similarly, the rate of user $2$ is simplified to
\begin{align}
R_2 &\leq \min\{R_K,0.5\log(1 + SNR)\}, \\
&=\min\{R_K,0.5\log SNR\}+ \mathcal{O}(1). \label{eq:godfachscheme6}
\end{align}
Hence, the achievable SGDOF for user~$2$ becomes,
\begin{align}
    d_2(\alpha,\gamma)&\leq \min\{\gamma,1\}. \label{eq:godfachscheme6b}
\end{align}
Using (\ref{eq:godfachscheme6a}) and (\ref{eq:godfachscheme6b}), the SGDOF region $\mathcal{D}^{I}(\alpha, \gamma)$ in the theorem is obtained.

The SGDOF region in  (\ref{eq:otp-time-sharing}) can be obtained by considering the following power allocation: $P_1=P$, and $P_2=\frac{1}{h_c^2}$. Under this power allocation, the rate of user $1$ simplifies to the following
\begin{align}
R_1&\leq 0.5\log\biggl(1+\frac{h_d^2P}{2}\biggr),\\ 
&=0.5\log(2 + SNR) + \mathcal{O}(1).  \label{eq:godfachscheme7}
\end{align}
For high SNR, the achievable SGDOF of user $1$ becomes
\begin{align}
    d_1(\alpha,\gamma)\leq 1.  \label{eq:godfachscheme7a}
\end{align}
Similarly, the rate of user $2$ is simplified to
\begin{align}
R_2&\leq \min\biggl\{R_K,0.5\log\biggl(1+\frac{h_d^2}{h_c^2}\biggr)\biggr\},\\
&=\min \lcb R_K,0.5\log SNR  - 0.5\log INR + \mathcal{O}(1)\rcb,  \label{eq:godfachscheme8}
\end{align}
Hence, the achievable SGDOF of user $2$ reduces to,
\begin{align}
    d_2(\alpha,\gamma)&\leq \min\{\gamma,1-\alpha\}. \label{eq:godfachscheme7b}
\end{align}
Using (\ref{eq:godfachscheme7a}) and (\ref{eq:godfachscheme7b}), the SGDOF region $\mathcal{D}^{II}(\alpha, \gamma)$ in the theorem is obtained.
\end{proof}
\section{Numerical Results}
\begin{figure*}[t]
	\centering
	\begin{subfigure}[b]{0.32\textwidth}
		\includegraphics[width=1\textwidth]{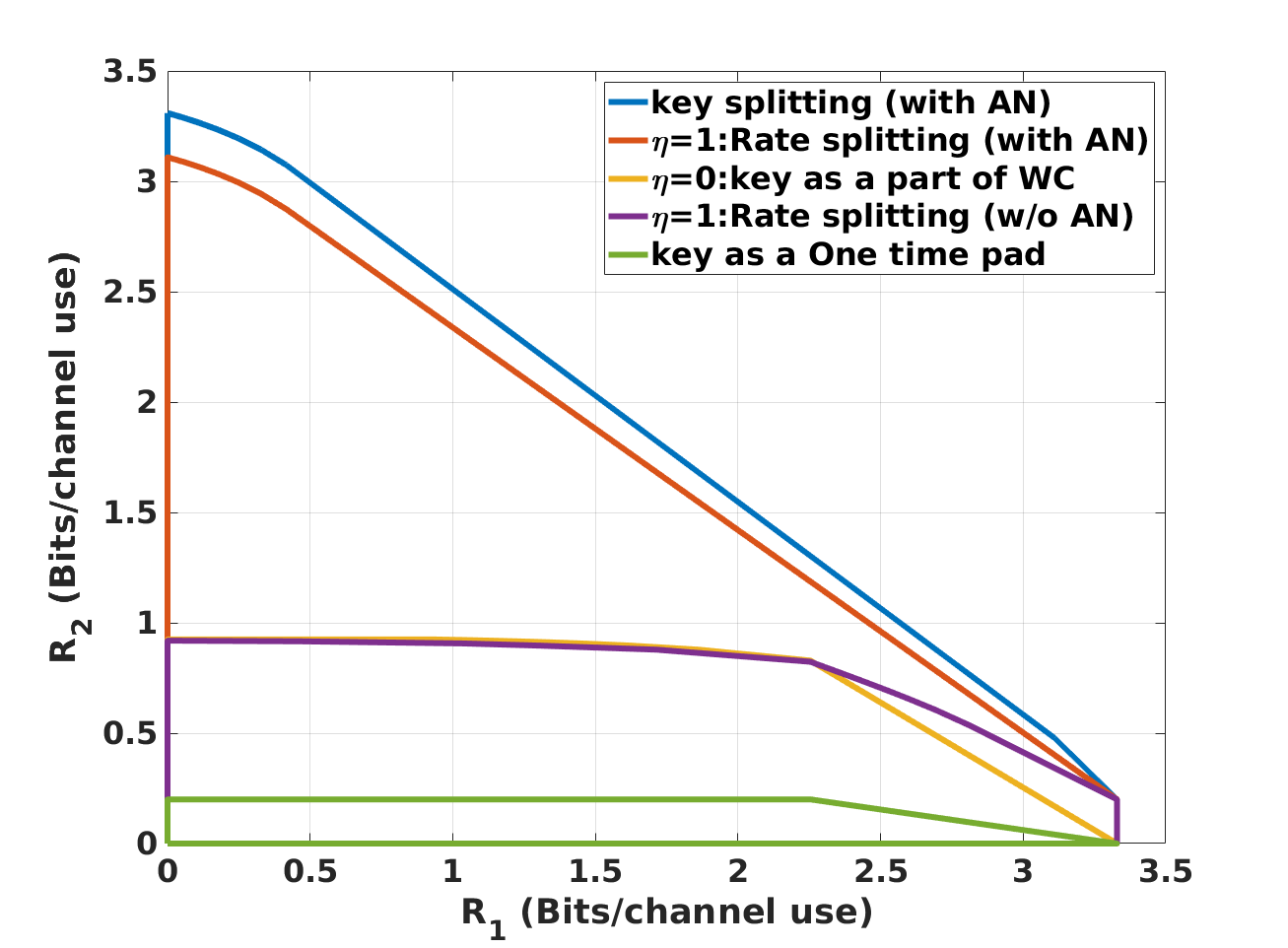}
		\caption{$h_{21}=0.6$ and $R_K=0.2$}
	\end{subfigure}
	\begin{subfigure}[b]{0.32\textwidth}
		\includegraphics[width=1\textwidth]{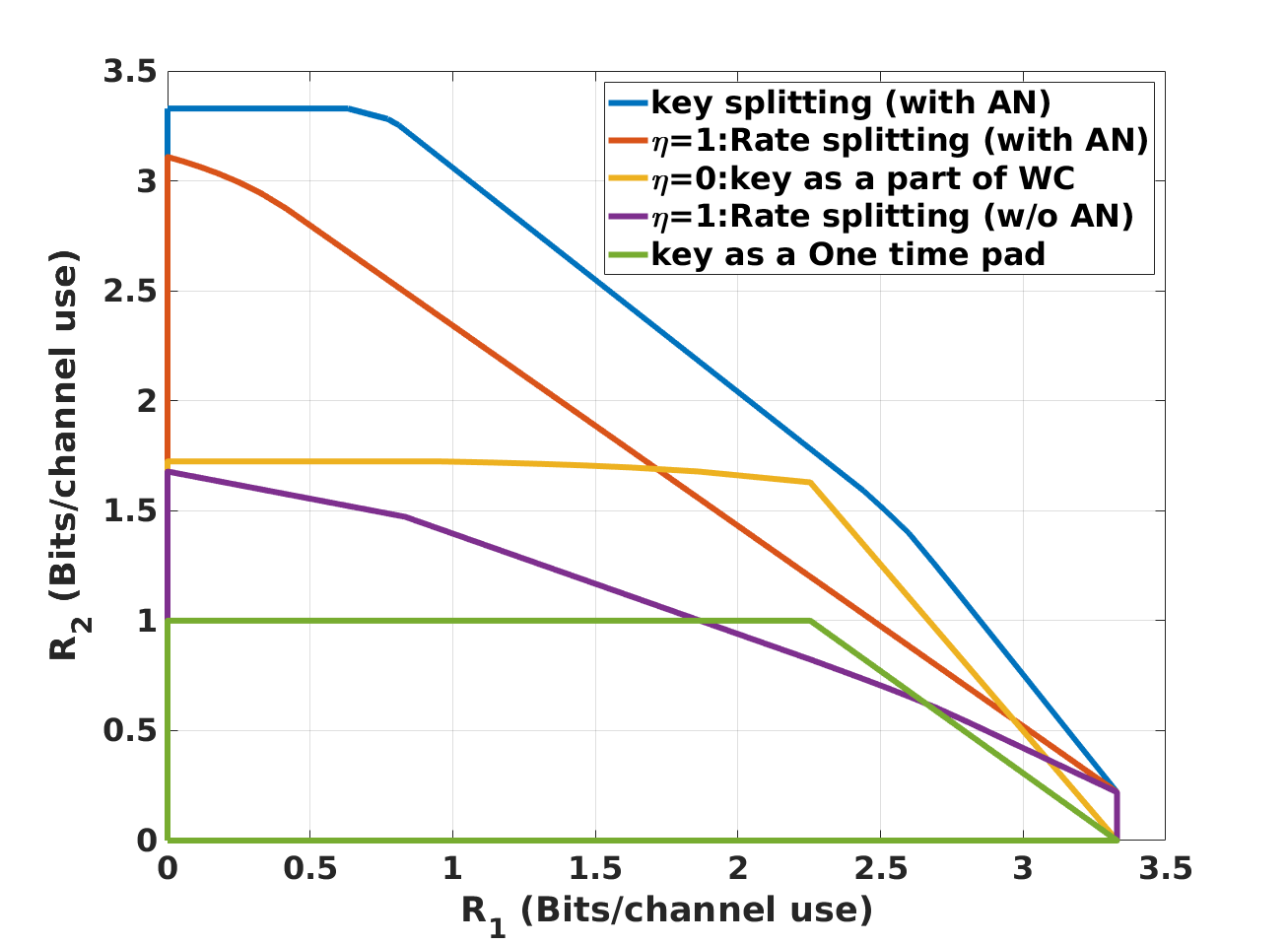}
		\caption{$h_{21}=0.6$ and $R_K=1$}
	\end{subfigure}
	\begin{subfigure}[b]{0.32\textwidth}
		\includegraphics[width=1\textwidth]{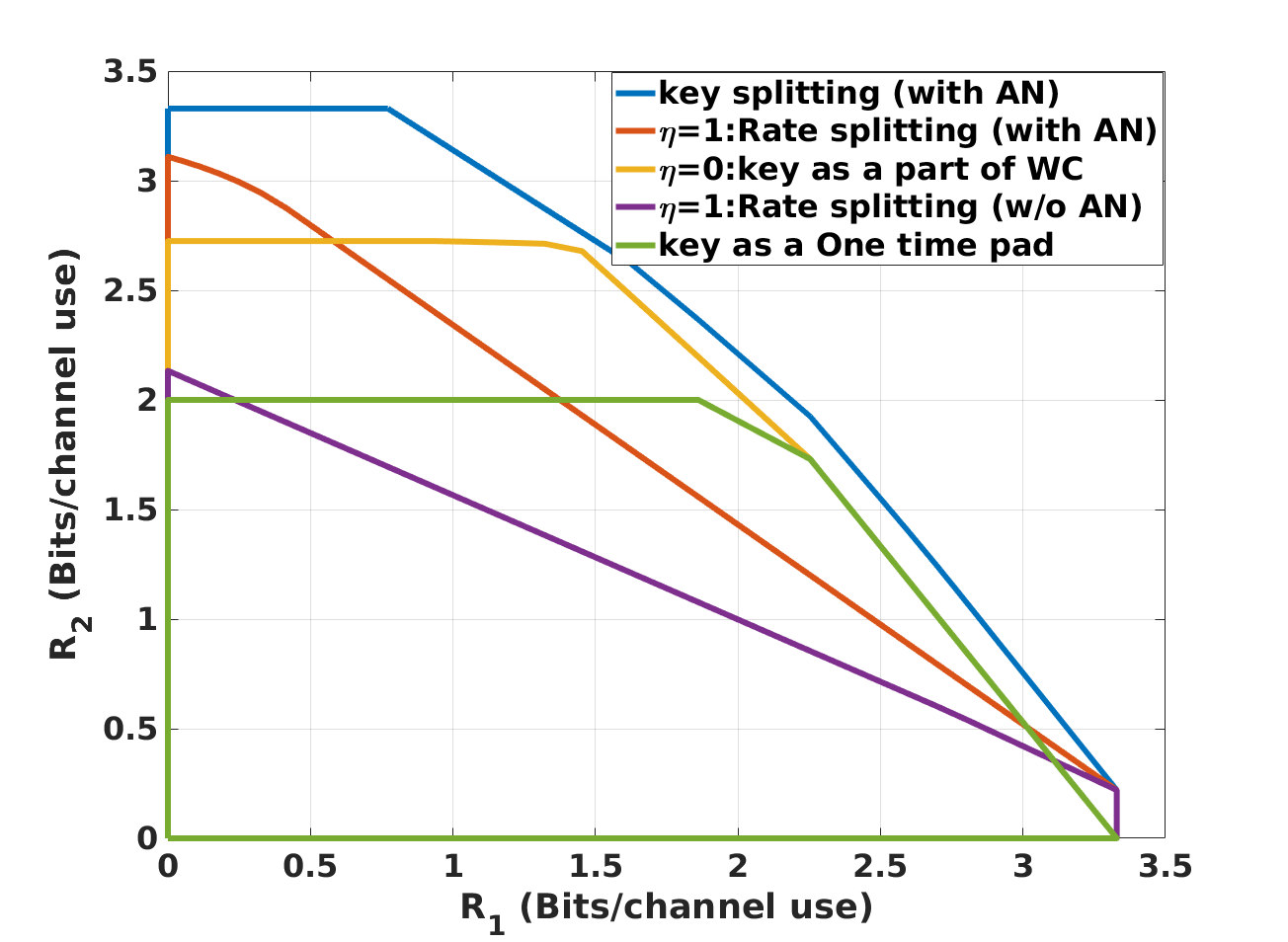}
		\caption{$h_{21}=0.6$ and $R_K=2$}
	\end{subfigure}
	\centering
	\begin{subfigure}[b]{0.32\textwidth}
		\includegraphics[width=1\textwidth]{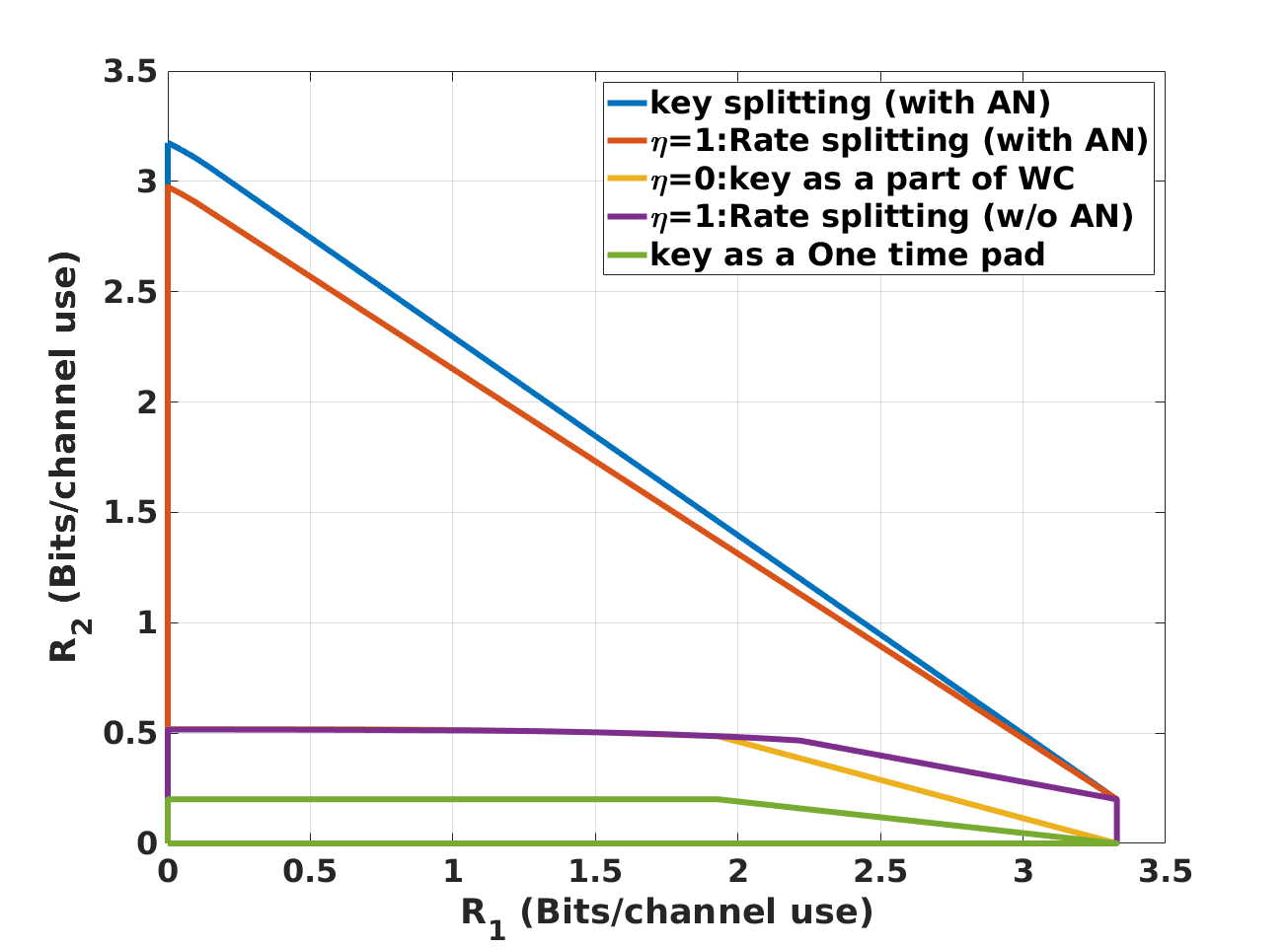}
		\caption{$h_{21}=0.8$ and $R_K=0.2$}
	\end{subfigure}
	\begin{subfigure}[b]{0.32\textwidth}
		\includegraphics[width=1\textwidth]{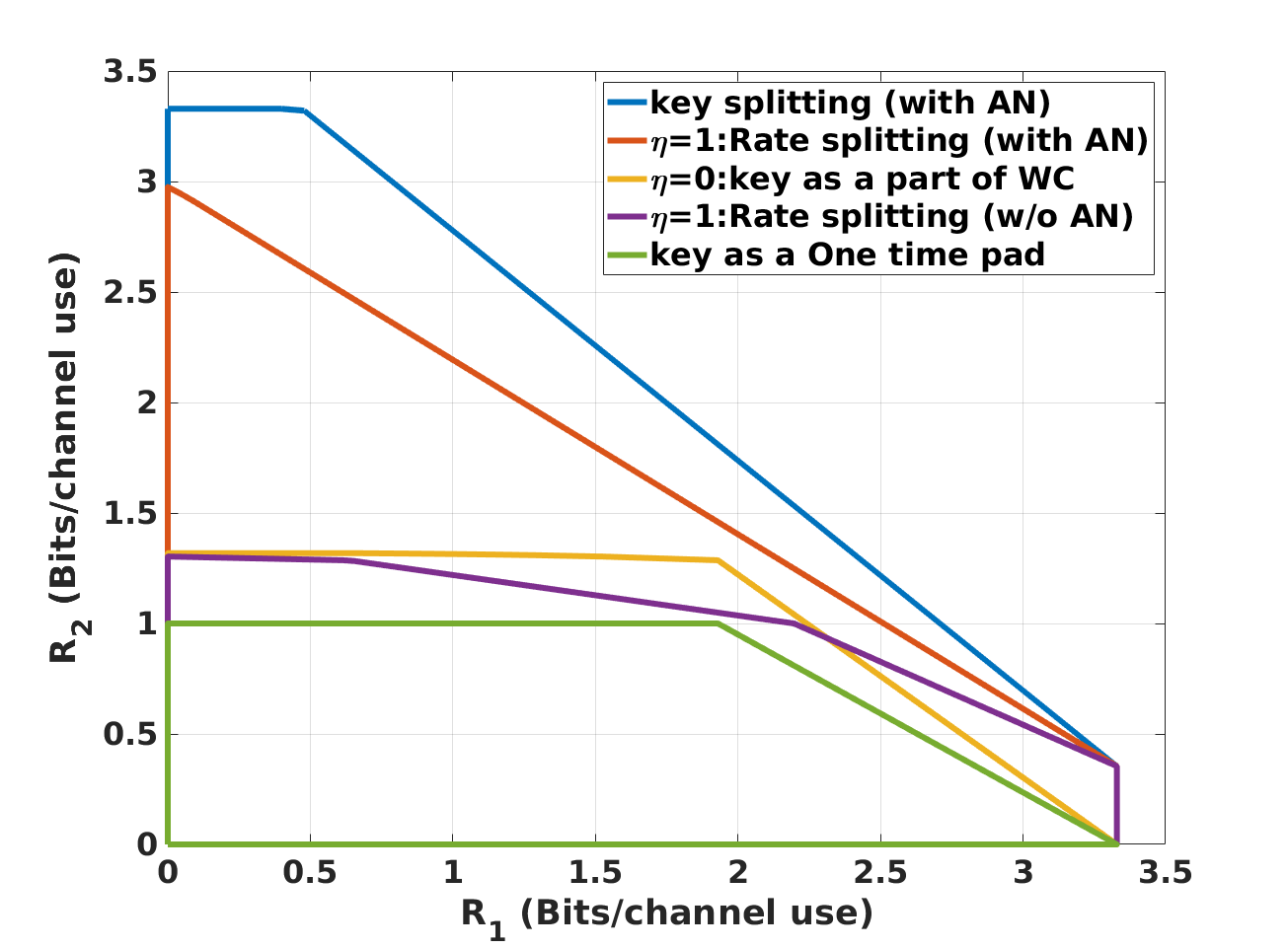}
		\caption{$h_{21}=0.8$ and $R_K=1$}
	\end{subfigure}
	\begin{subfigure}[b]{0.32\textwidth}
		\includegraphics[width=1\textwidth]{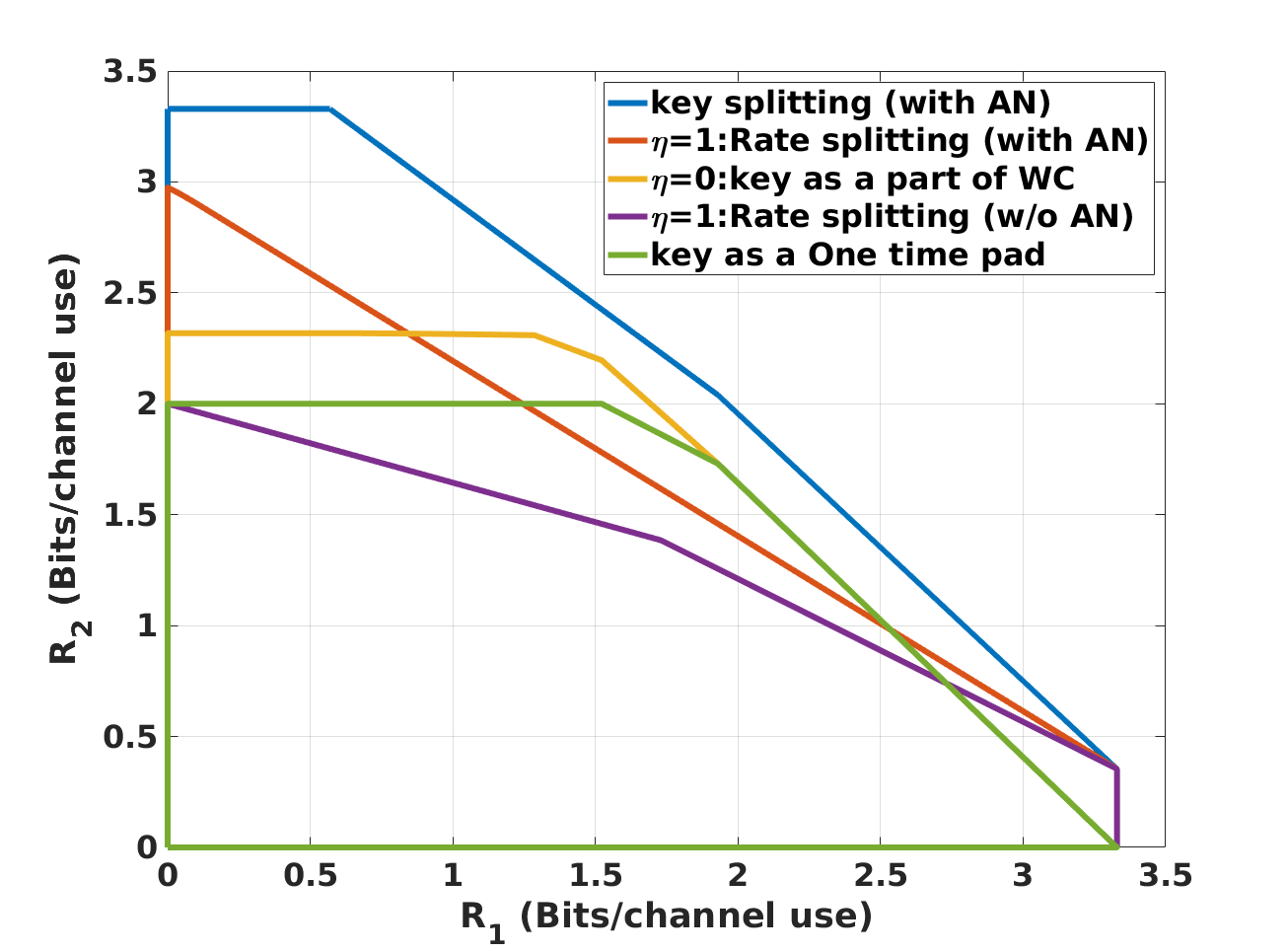}
		\caption{$h_{21}=0.8$ and $R_K=2$}
	\end{subfigure}
	\caption{Achievable rate regions with $P_1=P_2=100$ and $h_{11}=h_{22}=1$ for weak/moderate interference regime. In the legend WC stands for wiretap coding}\label{fig:achievable rate regions for weak/moderate interference regime}
\end{figure*}
\begin{figure*}[!]
	\centering
	\begin{subfigure}[b]{0.32\textwidth}
		\includegraphics[width=1\textwidth]{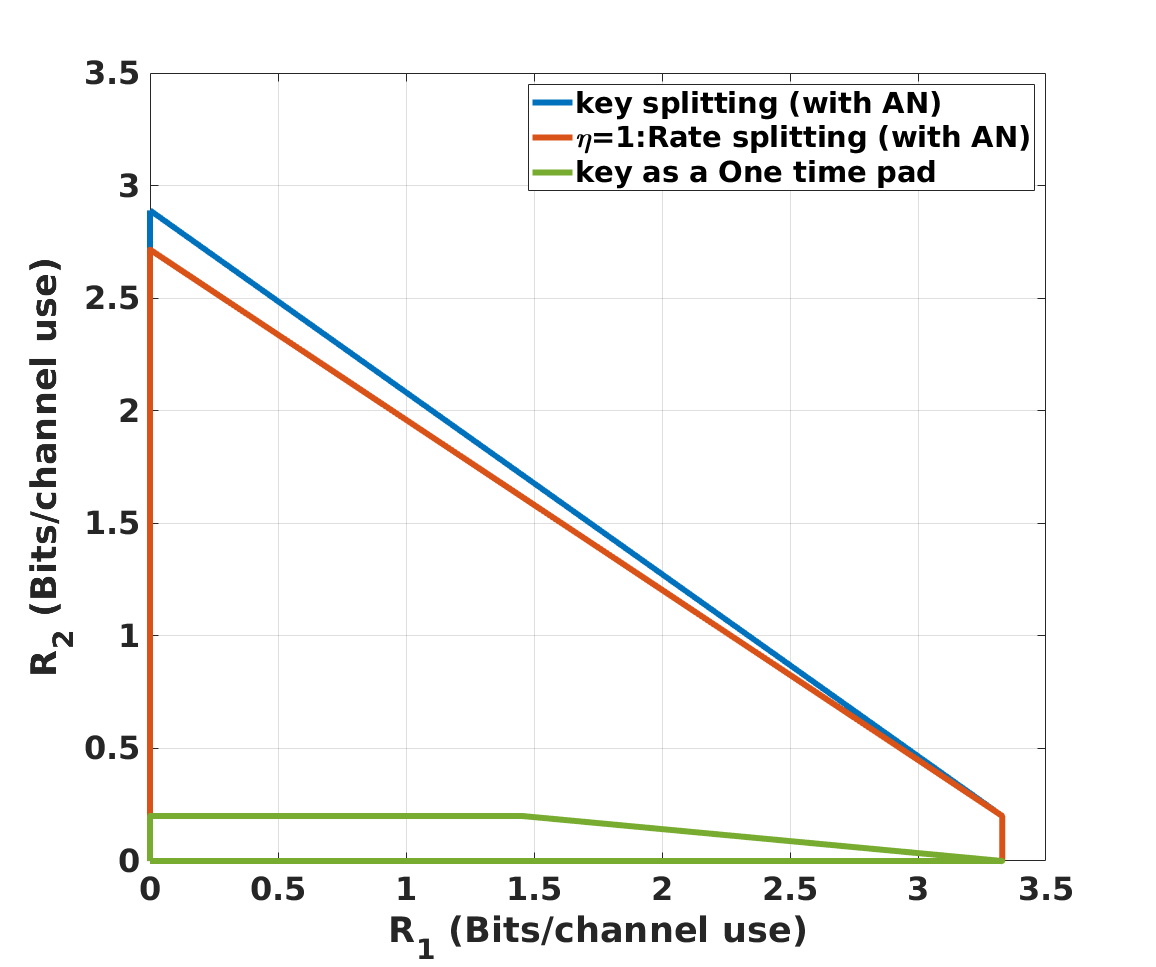}
		\caption{$h_{21}=1.2$ and $R_K=0.2$}
	\end{subfigure}
	\begin{subfigure}[b]{0.32\textwidth}
		\includegraphics[width=1\textwidth]{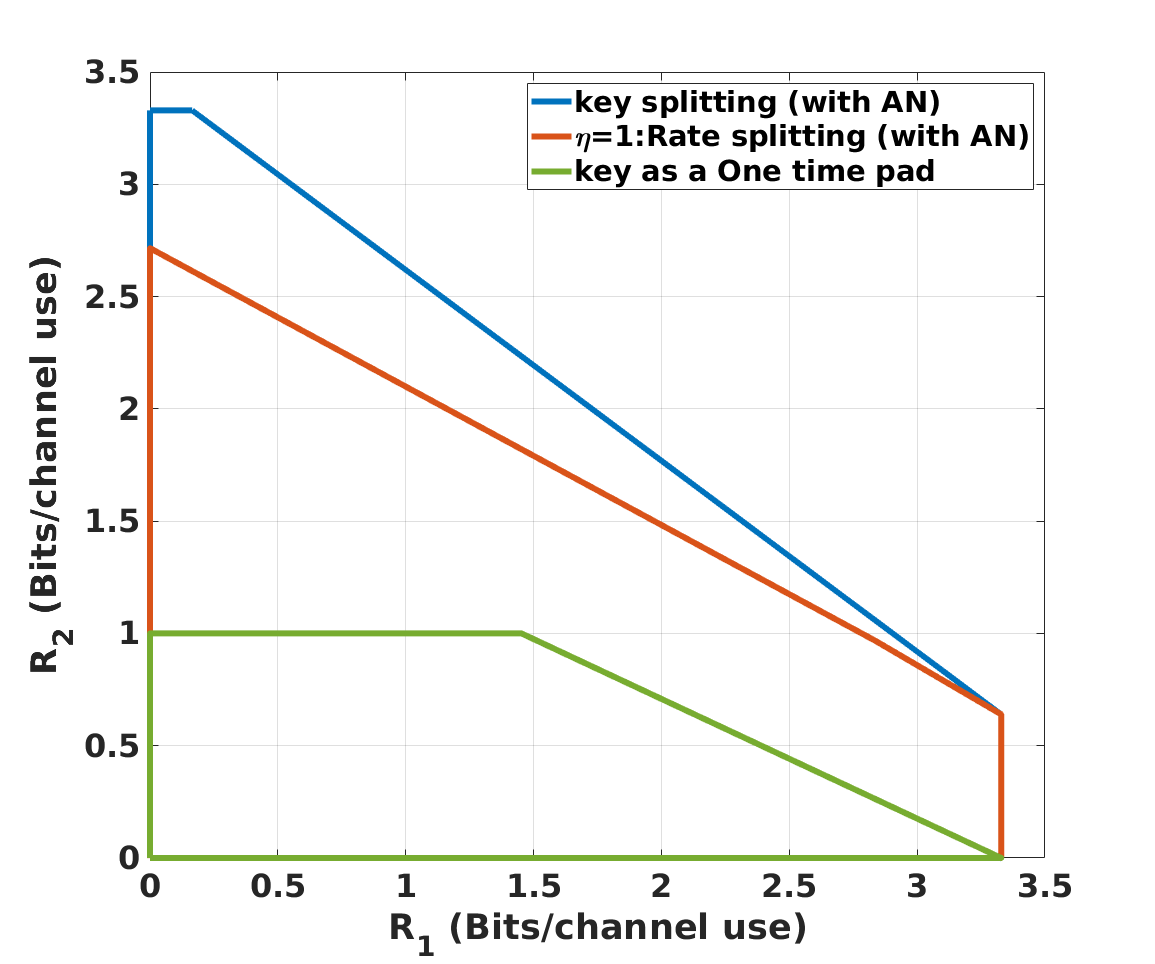}
		\caption{$h_{21}=1.2$ and $R_K=1$}
	\end{subfigure}
	\begin{subfigure}[b]{0.32\textwidth}
		\includegraphics[width=1\textwidth]{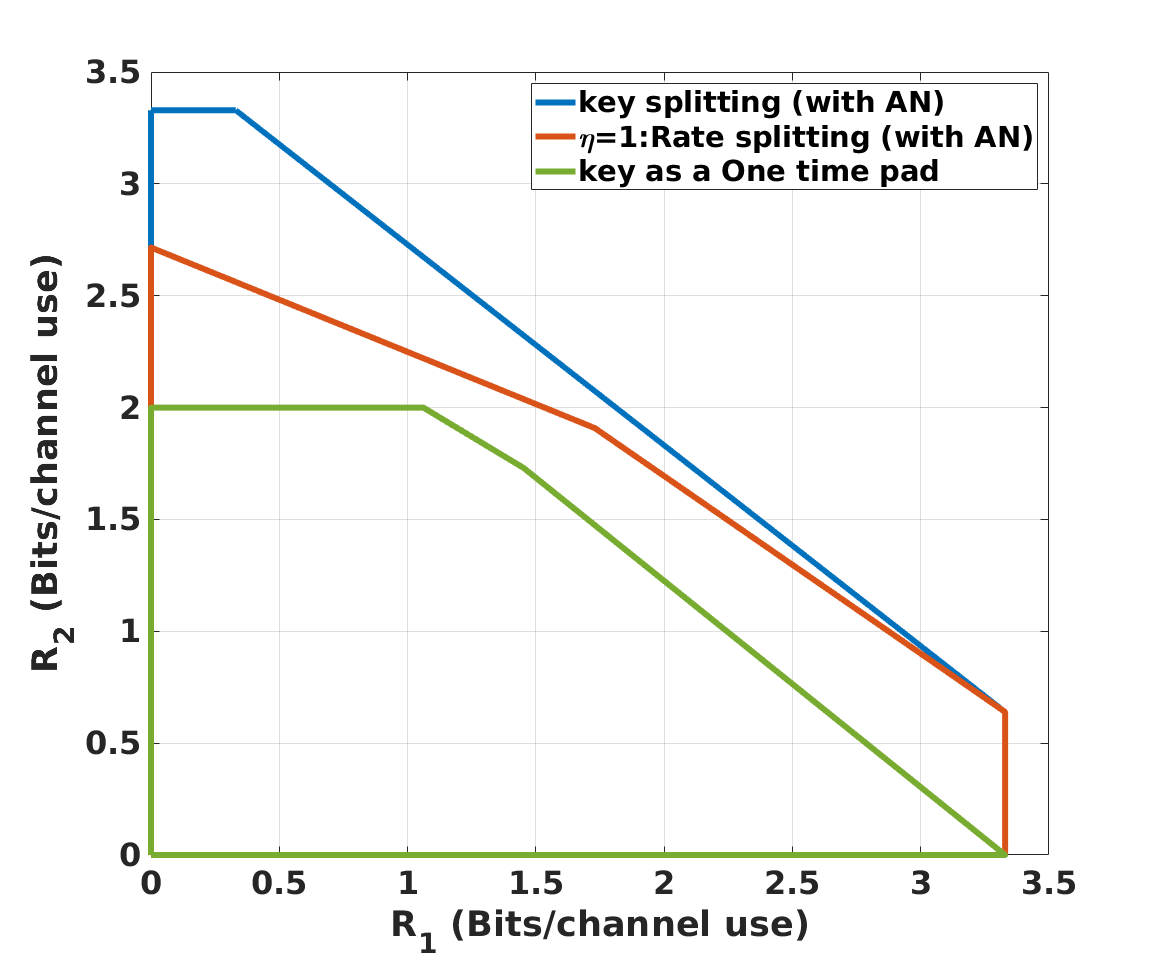}
		\caption{$h_{21}=1.2$ and $R_K=2$}
	\end{subfigure}
	\caption{Achievable rate regions with $P_1=P_2=100$ and $h_{11}=h_{22}=1$ for high interference regime.}\label{fig:achievable rate regions for high interference regime}
\end{figure*}
This section presents the performance of the various schemes and effectiveness of the outer bounds under different channel conditions and key rates. In Figs~\ref{fig:achievable rate regions for weak/moderate interference regime} and \ref{fig:achievable rate regions for high interference regime}, the achievable rates regions are plotted for different schemes for the weak/moderate interference and high interference regime, respectively. The results stated in Corollaries~\ref{corollary-key splitting}-\ref{corollary-key is used as a part of wiretap coding} and Theorem~\ref{theorem-key as a one time pad} are labelled as \emph{Key Splitting (with AN)}, \emph{Rate Splitting (with AN)}, \emph{Key as a part of WC}, and \emph{Key as a one time pad}, respectively. When the secret key is used only to encode the confidential message without AN transmission is also plotted using the result from Corollary~\ref{corollary-rate splitting approach} and this is labelled as \emph{Rate splitting (w/o AN)}. This also helps to explore the need of artificial noise transmission when users have access to shared key of finite rate.

For Fig~\ref{fig:achievable rate regions for weak/moderate interference regime}, both the direct channel gains are set as $h_{11} = h_{22} = 1$ and two cases are considered for the interfering link $(h_{21} = 0.6$ and $h_{21} = 0.8)$.  It can be observed that the achievable scheme with key rate splitting outperforms all the schemes for different values of $R_K$. Hence, using the entire key for encoding the common confidential message $(\eta = 1)$ or a part of the wiretap coding $(\eta = 0)$ can be highly sub-optimal. On the other hand, splitting the key rate into two parts allows to decode some part of the interference at the unintended receiver without violating the secrecy constraint and at the same time, it allows to reduce the loss in rate due to stochastic encoding. Hence, splitting the key rate helps to achieve better performance in comparison to other schemes. One can also note that when user~$2$ achieves the maximum rate, user~$1$ can achieve non-zero rate for the scheme based on key-splitting/rate splitting (Corollaries~ \ref{corollary-key splitting} and  \ref{corollary-rate splitting approach}). The scheme based on one-time pad, where the receiver~$1$ treats interference as noise can be highly suboptimal and the performance is primarily limited by the key rate. For the achievable schemes based on one-time pad or when key is used as a part of wiretap coding, when user~$1$ achieves the maximum rate, user~$2$ cannot achieve non-zero rate in contrast to other schemes. This illustrates the benefit of decoding some part of the interference without violating the secrecy constraint. It can also be observed that with increase in the strength of the cross-channel coefficient, there is shrinkage in the achievable rate region. 

In Fig~\ref{fig:achievable rate regions for high interference regime}, the achievable rate regions for different schemes are plotted for the high interference regime $h_{11} = h_{22} = 1$ and $h_{21} = 1.2$. The results without artificial noise transmission (Corollary~\ref{corollary-rate splitting approach} with $\lambda_1 = 1$) and when key is used as a part of the wiretap coding (Corollary~\ref{corollary-key is used as a part of wiretap coding}) are not considered as these schemes are not applicable for the high-interference regime. Note that to send private message using wiretap coding, it is required that the interfering channel should be weaker than the direct channel.  From the plot, it can be seen that the achievable scheme based on key rate splitting with artificial noise transmission performs the best among the schemes even in the high interference regime. This shows the importance of splitting the key rate and superposition coding  in case of interference limited scenarios.

\begin{figure*}[!]
	\centering
	\begin{subfigure}[b]{0.49\textwidth}
		\includegraphics[width=1\textwidth]{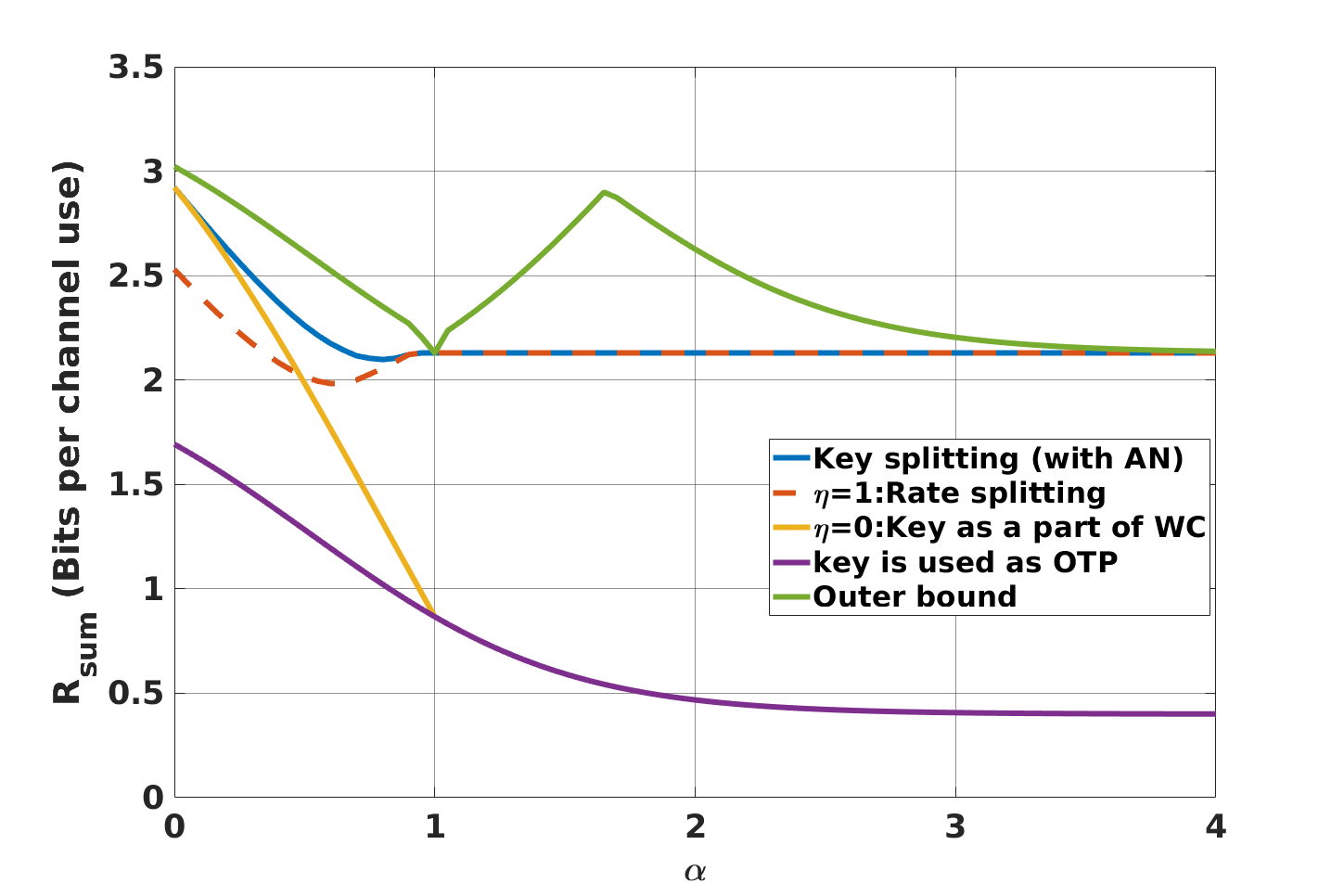}
		\caption{$R_K=0.4$}
	\end{subfigure}
	\begin{subfigure}[b]{0.49\textwidth}
		\includegraphics[width=1\textwidth]{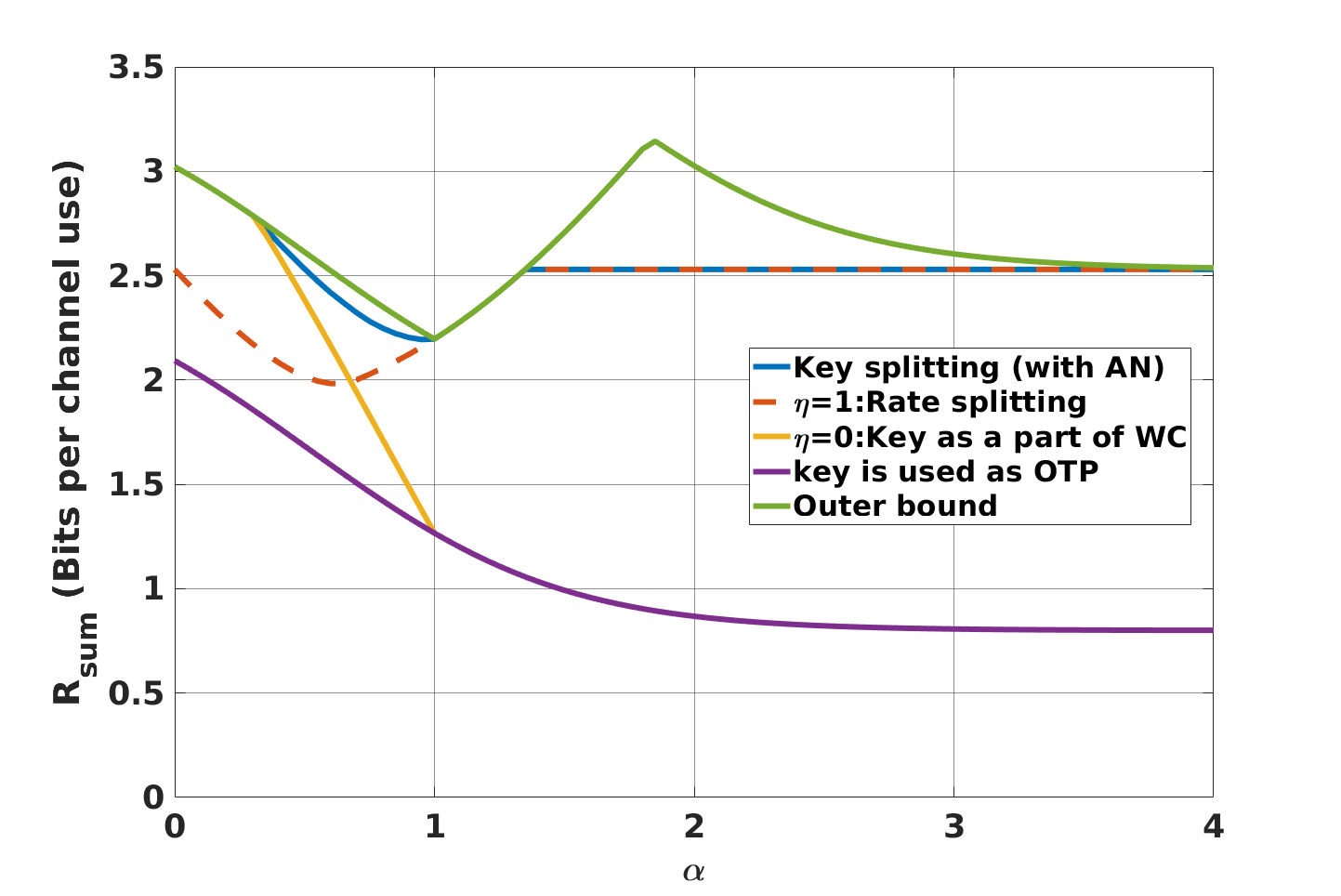}
		\caption{$R_K=0.8$}
	\end{subfigure}
	\caption{Sum rate of Z-IC for different schemes with $P_1=P_2=10$, $h_{11}=h_{22}=1$ and varying $\alpha$. In the legend OTP stands for one time pad.}\label{fig:sumrate vs alpha}
\end{figure*}
\begin{figure*}[!]
	\centering
	\begin{subfigure}[b]{0.49\textwidth}
		\includegraphics[width=1\textwidth]{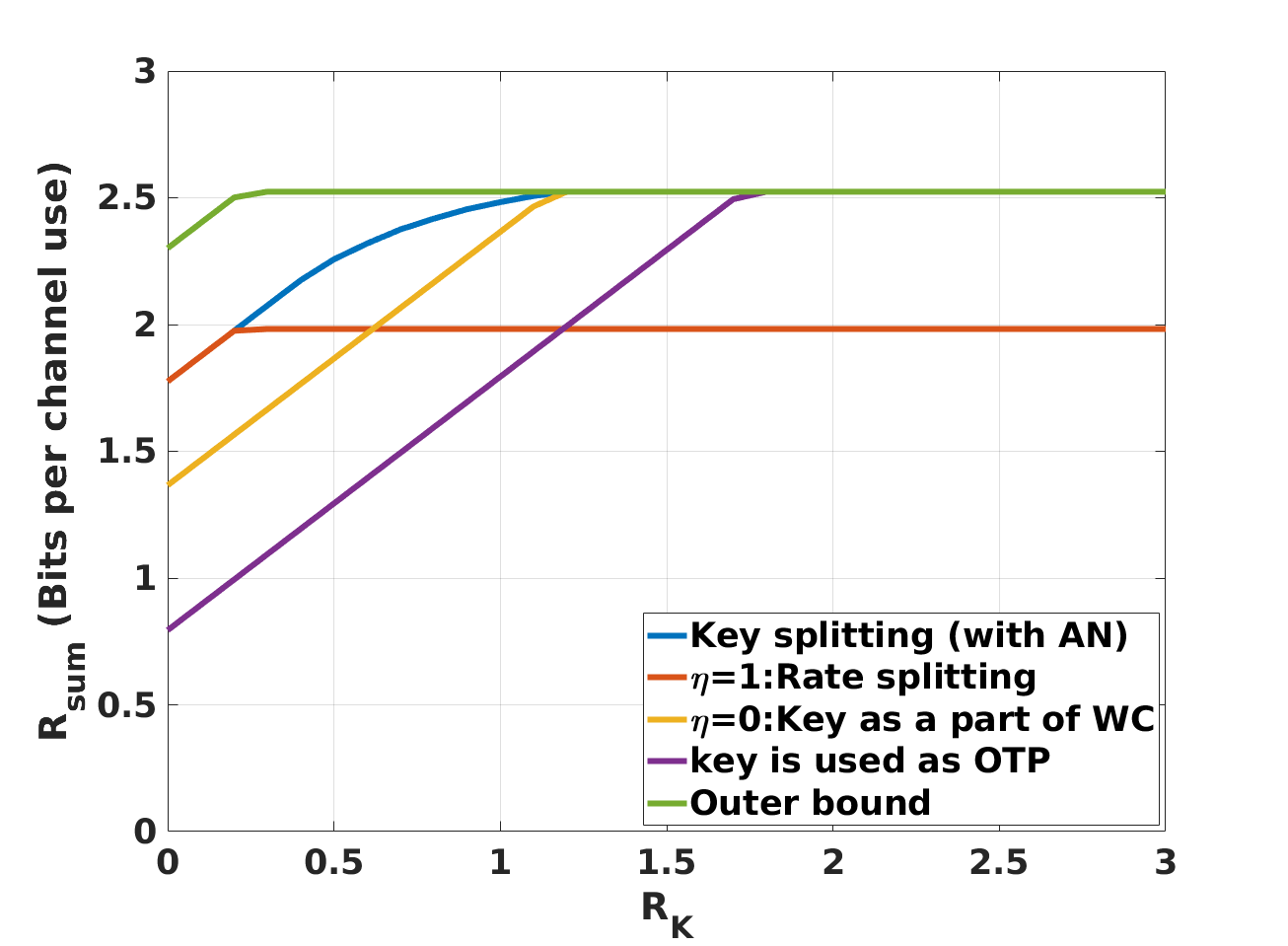}
		\caption{$\alpha=0.6$}
	\end{subfigure}
	\begin{subfigure}[b]{0.49\textwidth}
		\includegraphics[width=1\textwidth]{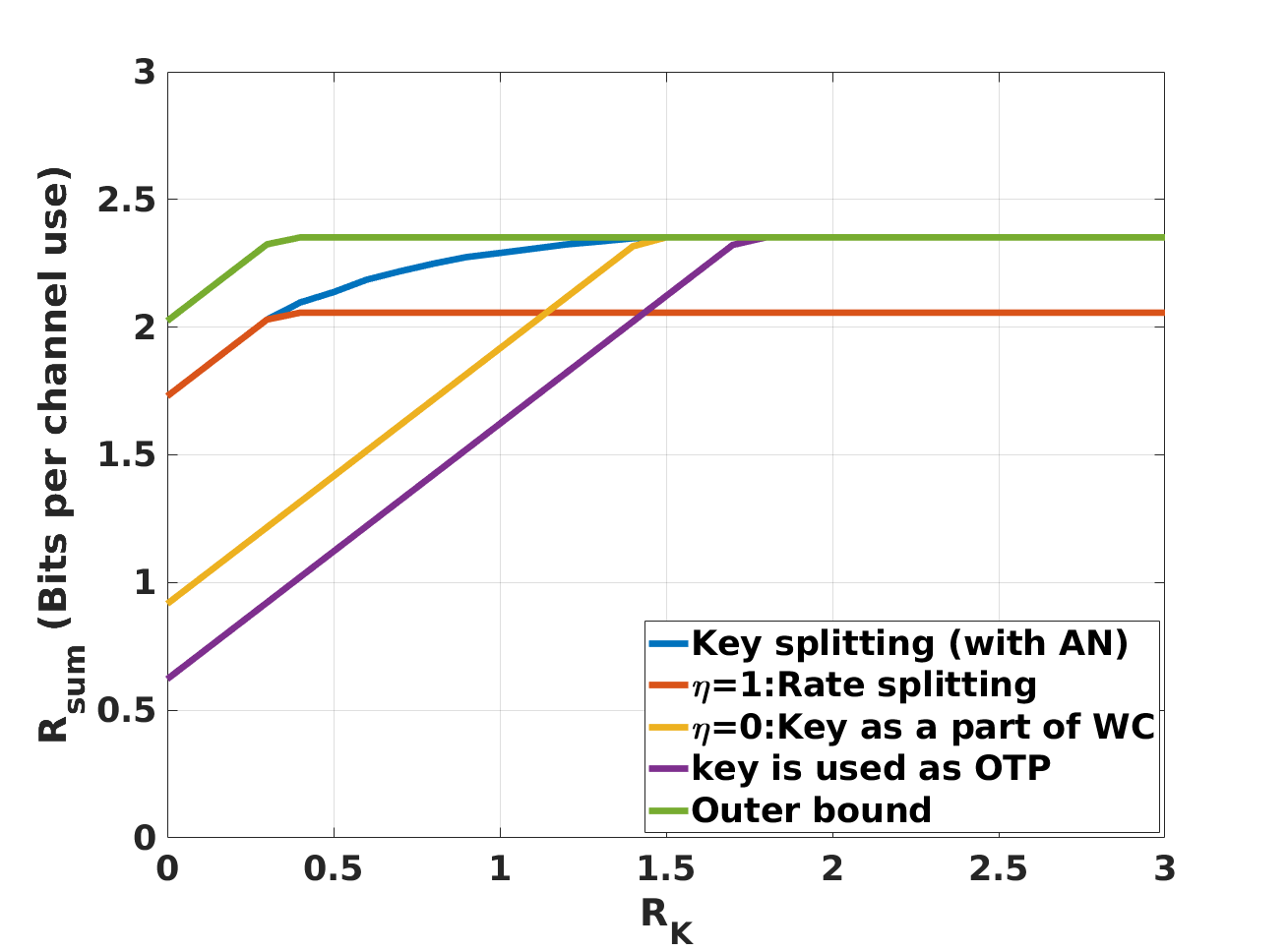}
		\caption{$\alpha=0.8$}
	\end{subfigure}
	\caption{Sum rate of Z-IC for different schemes with $P_1=P_2=10$, $h_{11}=h_{22}=1$ and varying $R_K$.}\label{fig:sumrate vs R_K}
\end{figure*}
In Fig. \ref{fig:sumrate vs alpha}, the achievable sum-rates are plotted for different schemes against $\alpha = \frac{\log h_c^2 P}{ \log h_d^2 P}$, where $h_d = h_{11}= h_{22}$ and $h_c = h_{21}$.  Along with the achievable sum rate, outer bound on the sum rate is also plotted. The outer bound is obtained by taking minimum of outer bounds in Theorem~\ref{Outer bound for weak/moderate interference regime}, Theorem~\ref{Outer bound for high interference regime} and  outer bounds without secrecy constraint in \cite{etkin-TIT-2008}. When $\alpha \geq 1$, outer bound in Theorem~\ref{Outer bound for weak/moderate interference regime} is not taken into consideration. The achievable sum rate for the scheme in Corollary~\ref{corollary-key is used as a part of wiretap coding} is plotted for the weak/moderate interference regime $(0 \leq \alpha \leq 1)$ only as this scheme is not applicable for $\alpha > 1$ as mentioned earlier. For the sum rate plot, both the transmitters use the full-power and no power control is used. When $R_K = 0.4$, the achievable scheme based on key rate splitting performs best among all the schemes and in the high interference regime $(\alpha \geq 1)$, the sum rate remains constant as the performance is limited by the key-rate. When the key rate is increased to $R_K = 0.8$, the scheme based on the key rate splitting performs best in the moderate interference regime $(0.5 \leq \alpha \leq  1)$. It can be seen that the achievable sum rate for the scheme based on key-rate splitting (Corollary~\ref{corollary-key splitting}) and rate splitting (Corollary~\ref{corollary-rate splitting approach}) increases with increase in the value of $\alpha$. In this case, using the key or some part of the key as a part of wiretap coding does not help from the sum-rate perspective.  For some ranges of $\alpha$, the schemes are found to be optimal as they coincide with the outer bound.
\begin{figure*}[!]
	\centering
	\begin{subfigure}[b]{0.49\textwidth}
		\includegraphics[width=1\textwidth]{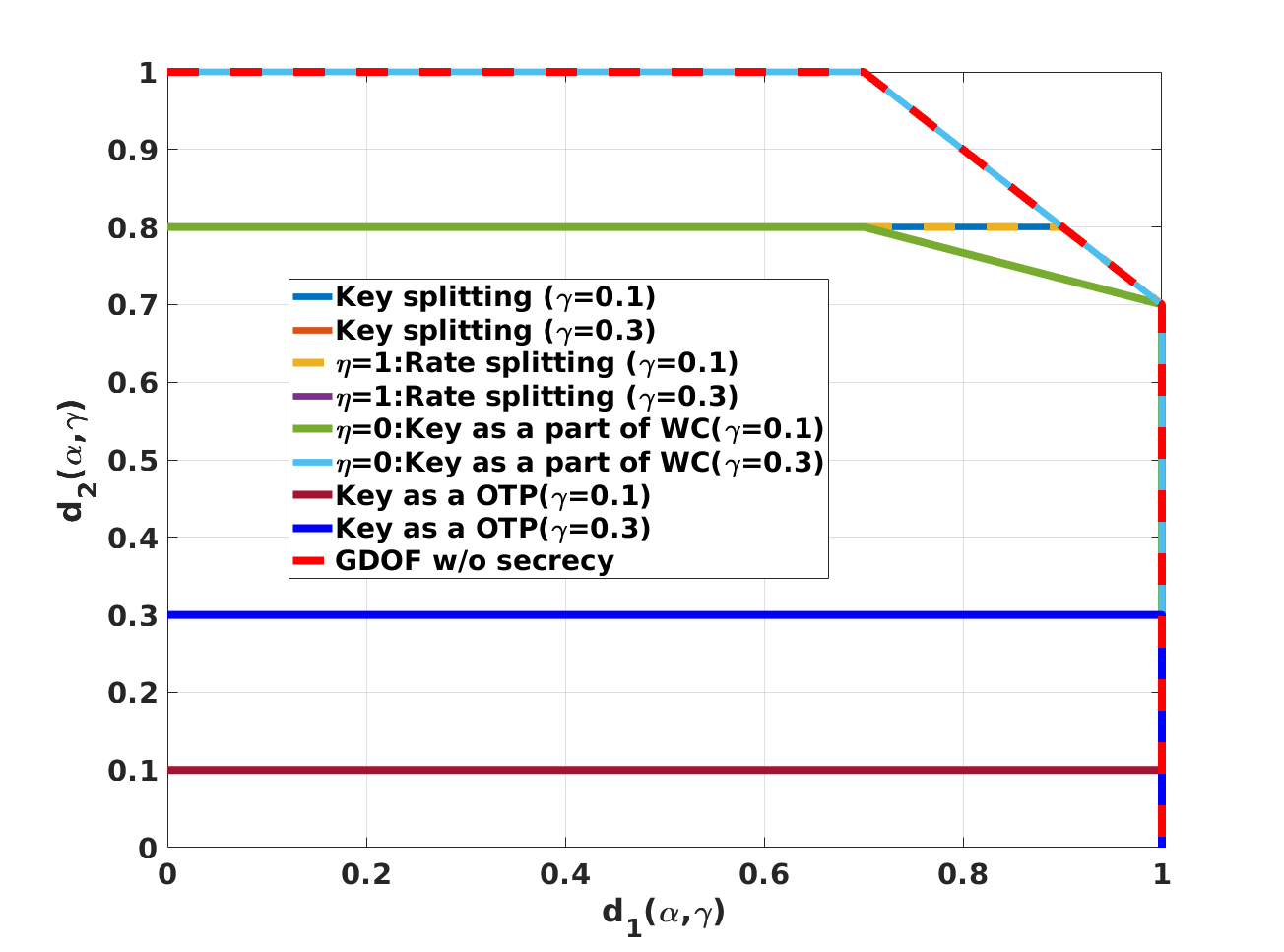}
		\caption{$\alpha=0.3$}
	\end{subfigure}
	\begin{subfigure}[b]{0.49\textwidth}
		\includegraphics[width=1\textwidth]{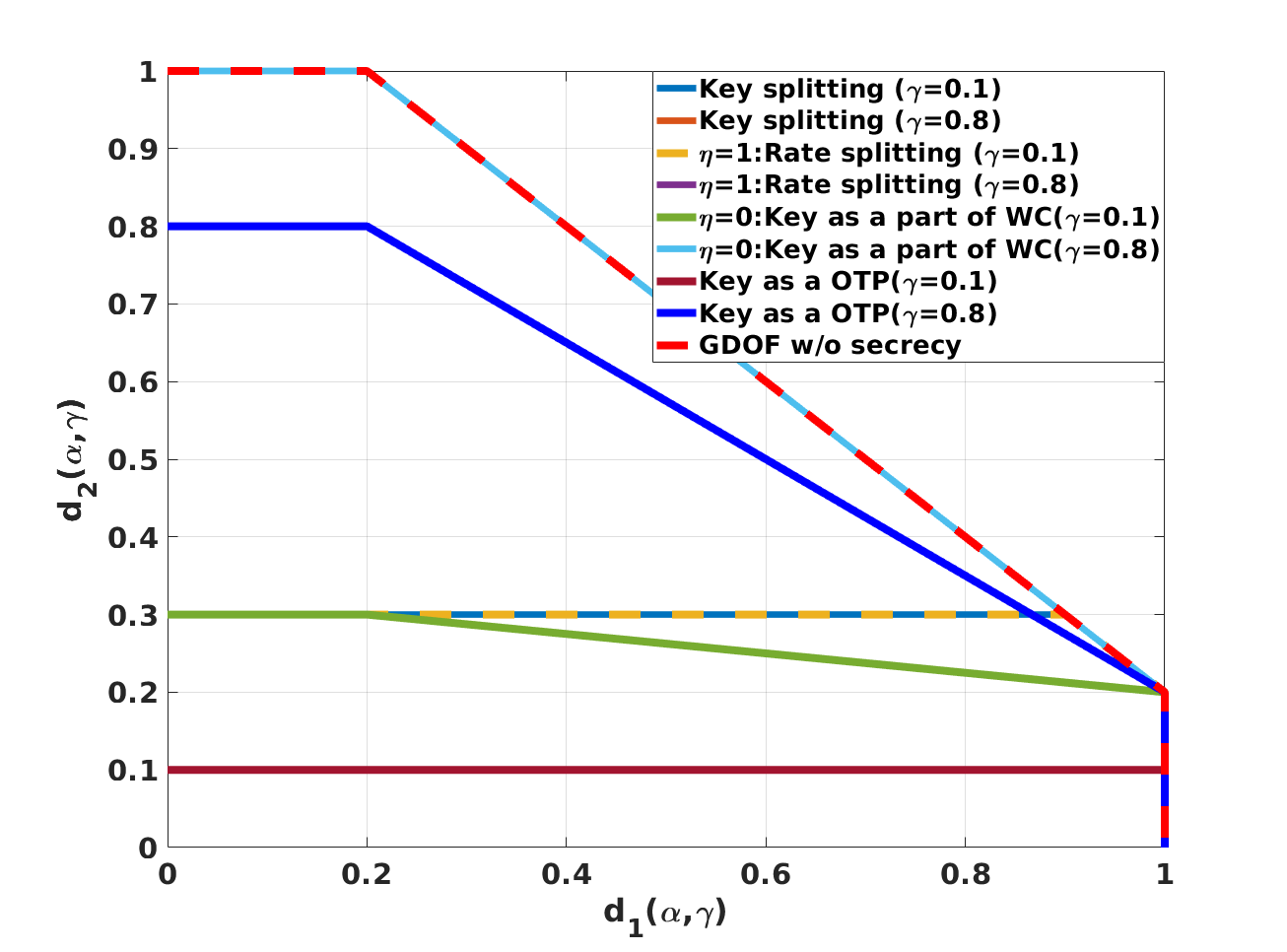}
		\caption{$\alpha=0.8$}
	\end{subfigure}
	\caption{GDOF regions for different schemes}.\label{fig:GDOF regions}
\end{figure*}

In Fig. \ref{fig:sumrate vs R_K}, the achievable sum-rates for different schemes are plotted against $R_K$ along with outer bound on the sum rate for different values of $\alpha$. From the plots, it can be observed that as the key rate increases from $0$ to $0.3$, the achievable scheme based on key rate splitting and rate splitting perform the same. With further increase in the value of $R_K$, the achievable scheme based on key rate splitting perform better compared to other schemes. This is due to the fact that for a given $\alpha$, as $R_K$ increases, using the key only for encoding the common confidential message can be sub-optimal as the rate associated with the common confidential message is limited by the amount of interference that can be decoded by receiver~$1$. Using some part of the key as a part of wiretap coding, can allow to send more private messages and hence, the sum rate increases. It is interesting to note that achievable scheme based on one-time pad can meet the outer bound but for higher value of key rate in comparison to other schemes.  When the value of $\alpha$ is increased to $0.8$, the sum rate of all achievable schemes decreases in comparison to $\alpha = 0.6$, but, the achievable scheme based on key rate splitting  outperforms other schemes.

In Fig.~\ref{fig:GDOF regions}, the SGDOF region in Theorems~\ref{gdof-key splitting}, \ref{theorem-gdof of key as a one time pad} and Corollaries \ref{gdof-rate splitting}, \ref{gdof-key is used as a part of WC} are plotted for different values of $\alpha$ and $\gamma$. It can be noticed that when $\gamma = 0.1$, the key splitting in Theorem~\ref{gdof-key splitting} and rate splitting in Corollary \ref{gdof-rate splitting} allows both the users to achieve higher GDOF in comparison to other schemes. For the considered power allocation scheme, there is no benefit of key-rate splitting in comparison to the rate splitting. It is interesting to note that when $\gamma = \alpha = 0.3~(\text{or }0.8)$, the key splitting scheme [Theorem \ref{gdof-key splitting}], rate splitting scheme [Corollary ~\ref{gdof-rate splitting}] and when the key is used as a fictitious message [Corollary ~\ref{gdof-key is used as a part of WC}] achieves the optimal GDOF as it coincides with the SGDOF region of the Z-IC without secrecy constraint \cite{etkin-TIT-2008}. When the key rate scales with $\alpha$, i.~e., $\gamma = \alpha$, there is no loss in the SGDOF performance due to the secrecy constraint at the receiver. From the GDOF regions, it can be seen that when user~$1$ achieves the maximum GDOF of 1, user~$2$ can achieve non-zero GDOF and vice-versa with the use of key. When SNR/INR are not very high, the benefits of splitting the key rate in improving the rate region can be seen from Figs.~\ref{fig:achievable rate regions for weak/moderate interference regime} - \ref{fig:sumrate vs R_K}.

\section{Conclusion}
 This paper considered a $2$-user Gaussian Z-IC with shared key of finite rate with secrecy constraint at receiver~$1$. The work proposes novel achievable scheme for the considered model where the achievable schemes differ from each other based on how the key has been used in the encoding process. The derived results illustrate the usefulness of key-rate splitting where some part of the key is used for the common confidential message and remaining part of the key is used in the wiretap coding. Outer bounds on the sum rate and secrecy rate of transmitter~$2$ are obtained using the secrecy constraint at receiver~$1$ and shared key of finite rate between the users. The work also characterizes the SGDOF of various schemes under weak/moderate interference regime which in turn help to analyze the scaling behaviour of key rate with respect to underlying channel parameters. 
\bibliographystyle{IEEEtran}
\bibliography{references}

\begin{thebibliography}{10}
\providecommand{\url}[1]{#1}
\csname url@samestyle\endcsname
\providecommand{\newblock}{\relax}
\providecommand{\bibinfo}[2]{#2}
\providecommand{\BIBentrySTDinterwordspacing}{\spaceskip=0pt\relax}
\providecommand{\BIBentryALTinterwordstretchfactor}{4}
\providecommand{\BIBentryALTinterwordspacing}{\spaceskip=\fontdimen2\font plus
\BIBentryALTinterwordstretchfactor\fontdimen3\font minus
  \fontdimen4\font\relax}
\providecommand{\BIBforeignlanguage}[2]{{%
\expandafter\ifx\csname l@#1\endcsname\relax
\typeout{** WARNING: IEEEtran.bst: No hyphenation pattern has been}%
\typeout{** loaded for the language `#1'. Using the pattern for}%
\typeout{** the default language instead.}%
\else
\language=\csname l@#1\endcsname
\fi
#2}}
\providecommand{\BIBdecl}{\relax}
\BIBdecl

\bibitem{shannon-bell-1949}
C.~Shannon, ``Communication theory of secrecy systems,'' \emph{Bell Syst. Tech.
  J.}, vol.~28, no.~4, pp. 656--715, Oct. 1949.

\bibitem{liu2004capacity}
N.~Liu and S.~Ulukus, ``On the capacity region of the gaussian {Z}-channel,''
  in \emph{IEEE Global Telecommunications Conference, 2004. GLOBECOM'04.},
  vol.~1.\hskip 1em plus 0.5em minus 0.4em\relax IEEE, 2004, pp. 415--419.

\bibitem{mohapatra2016secrecy}
P.~Mohapatra, C.~R. Murthy, and J.~Lee, ``On the secrecy capacity region of the
  two-user symmetric {Z} interference channel with unidirectional transmitter
  cooperation,'' \emph{IEEE Transactions on Information Forensics and
  Security}, vol.~12, no.~3, pp. 572--587, Oct. 2016.

\bibitem{li2008secrecy}
Z.~Li, R.~D. Yates, and W.~Trappe, ``Secrecy capacity region of a class of
  one-sided interference channel,'' in \emph{2008 IEEE International Symposium
  on Information Theory}.\hskip 1em plus 0.5em minus 0.4em\relax IEEE, 2008,
  pp. 379--383.

\bibitem{liu2009capacity}
N.~Liu and A.~J. Goldsmith, ``Capacity regions and bounds for a class of
  {Z}-interference channels,'' \emph{IEEE transactions on information theory},
  vol.~55, no.~11, pp. 4986--4994, 2009.

\bibitem{wyner-bell-1975}
A.~Wyner, ``The wire-tap channel,'' \emph{Bell Syst. Tech. J.}, vol.~54, no.~8,
  pp. 1334--1387, Oct. 1975.

\bibitem{csiszar1978broadcast}
I.~Csisz{\'a}r and J.~Korner, ``Broadcast channels with confidential
  messages,'' \emph{IEEE transactions on information theory}, vol.~24, no.~3,
  pp. 339--348, May 1978.

\bibitem{liu2008discrete}
R.~Liu, I.~Maric, P.~Spasojevic, and R.~D. Yates, ``Discrete memoryless
  interference and broadcast channels with confidential messages: Secrecy rate
  regions,'' \emph{IEEE Transactions on Information Theory}, vol.~54, no.~6,
  pp. 2493--2507, Jun 2008.

\bibitem{mohapatra-tit-2016}
P.~{Mohapatra} and C.~R. {Murthy}, ``On the capacity of the two-user symmetric
  interference channel with transmitter cooperation and secrecy constraints,''
  \emph{IEEE Transactions on Information Theory}, vol.~62, no.~10, pp.
  5664--5689, 2016.

\bibitem{etkin-TIT-2008}
R.~Etkin, D.~Tse, and H.~Wang, ``{Gaussian interference channel capacity to
  within one bit},'' vol.~54, no.~12, pp. 5534--5562, Dec. 2008.

\bibitem{xu-tifs-2013}
P.~{Xu}, Z.~{Ding}, and X.~{Dai}, ``Rate regions for multiple access channel
  with conference and secrecy constraints,'' \emph{IEEE Transactions on
  Information Forensics and Security}, vol.~8, no.~12, pp. 1961--1974, 2013.

\bibitem{fritschek-tifs-2019}
R.~{Fritschek} and G.~{Wunder}, ``On the gaussian multiple access wiretap
  channel and the gaussian wiretap channel with a helper: Achievable schemes
  and upper bounds,'' \emph{IEEE Transactions on Information Forensics and
  Security}, vol.~14, no.~5, pp. 1224--1239, 2019.

\bibitem{fayed-asilomer-2016}
A.~{Fayed}, T.~{Khattab}, and L.~{Lai}, ``Secret communication on the
  {Z}-channel with cooperative receivers,'' in \emph{2016 50th Asilomar
  Conference on Signals, Systems and Computers}, 2016, pp. 909--914.

\bibitem{yamamoto1997rate}
H.~Yamamoto, ``Rate-distortion theory for the shannon cipher system,''
  \emph{IEEE Transactions on Information Theory}, vol.~43, no.~3, pp. 827--835,
  May. 1997.

\bibitem{merhav2008shannon}
N.~Merhav, ``Shannon's secrecy system with informed receivers and its
  application to systematic coding for wiretapped channels,'' \emph{IEEE
  Transactions on Information Theory}, vol.~54, no.~6, pp. 2723--2734, Jun.
  2008.

\bibitem{kang2010wiretap}
W.~Kang and N.~Liu, ``Wiretap channel with shared key,'' in \emph{2010 IEEE
  Information Theory Workshop}, Aug/Sep. 2010, pp. 1--5.

\bibitem{ardestanizadeh2009wiretap}
E.~Ardestanizadeh, M.~Franceschetti, T.~Javidi, and Y.-H. Kim, ``Wiretap
  channel with secure rate-limited feedback,'' \emph{IEEE Transactions on
  Information Theory}, vol.~55, no.~12, pp. 5353--5361, Dec 2009.

\bibitem{schaefer2017secure}
R.~F. Schaefer, A.~Khisti, and H.~V. Poor, ``Secure broadcasting using
  independent secret keys,'' \emph{IEEE Transactions on Communications},
  vol.~66, no.~2, pp. 644--661, Oct. 2017.

\bibitem{sinha2018secrecy}
A.~Sinha, P.~Mohapatra, J.~Lee, and T.~Q. Quek, ``On the secrecy capacity of
  2-user gaussian interference channel with independent secret keys,'' in
  \emph{2018 International Symposium on Information Theory and Its Applications
  (ISITA)}, Oct. 2018, pp. 663--667.

\bibitem{salehkalaibar2010capacity}
S.~Salehkalaibar and M.~R. Aref, ``On the capacity region of the degraded {Z}
  channel,'' in \emph{2010 IEEE Information Theory Workshop}.\hskip 1em plus
  0.5em minus 0.4em\relax IEEE, 2010, pp. 1--5.

\bibitem{sason2004achievable}
I.~Sason, ``On achievable rate regions for the gaussian interference channel,''
  \emph{IEEE transactions on information theory}, vol.~50, no.~6, pp.
  1345--1356, 2004.

\bibitem{gamal1982capacity}
A.~Gamal and M.~Costa, ``The capacity region of a class of deterministic
  interference channels (corresp.),'' \emph{IEEE Transactions on information
  Theory}, vol.~28, no.~2, pp. 343--346, 1982.

\bibitem{zhou2012gaussian}
L.~Zhou and W.~Yu, ``Gaussian z-interference channel with a relay link:
  Achievability region and asymptotic sum capacity,'' \emph{IEEE transactions
  on information theory}, vol.~58, no.~4, pp. 2413--2426, 2012.

\bibitem{bagheri2010approximate}
H.~Bagheri, A.~S. Motahari, and A.~K. Khandani, ``The approximate capacity
  region of the gaussian {Z}-interference channel with conferencing encoders,''
  \emph{arXiv preprint arXiv:1005.1635}, 2010.

\bibitem{do2009achievable}
H.~T. Do, T.~J. Oechtering, and M.~Skoglund, ``An achievable rate region for
  the gaussian {Z}-interference channel with conferencing,'' in \emph{2009 47th
  Annual Allerton Conference on Communication, Control, and Computing
  (Allerton)}.\hskip 1em plus 0.5em minus 0.4em\relax IEEE, 2009, pp. 75--81.

\bibitem{bloch2013strong}
M.~R. Bloch and J.~N. Laneman, ``Strong secrecy from channel resolvability,''
  \emph{IEEE Transactions on Information Theory}, vol.~59, no.~12, pp.
  8077--8098, Dec. 2013.

\bibitem{cover2006elements}
T.~M. Cover and J.~A. Thomas, ``Elements of information theory 2nd edition
  (wiley series in telecommunications and signal processing),'' July. 2006.

\bibitem{el2011network}
A.~El~Gamal and Y.-H. Kim, \emph{Network information theory}.\hskip 1em plus
  0.5em minus 0.4em\relax Cambridge university press, 2011.

\end{thebibliography}
\end{document}